\documentclass[10pt,a4paper]{article}

\usepackage{amsmath,amsgen,latexsym}
\usepackage{amstext,amssymb,amsfonts,latexsym}
\usepackage{theorem}
\usepackage{pifont}
\usepackage{graphicx}

 \newcommand{\bs}{\bigskip}
 
 \newcommand{\n}{\noindent}
 \newcommand{\s}{\smallskip}
 \newcommand{\hs}[1]{\hspace*{ #1 mm}}
 \newcommand{\vs}[1]{\vspace*{ #1 mm}}

 \newcommand{\setempty}{\mathrm{\O}}
 
 \newcommand{\nat}{\mathbb{N}}
 
 \newcommand{\integer}{\mathbb{Z}}

 \newcommand{\etalc}{\textrm{et al.}}

 \newcommand{\CC}{{\cal C}}
 \newcommand{\FF}{{\cal F}}

 \newcommand{\LL}{{\cal L}}
 
 \newcommand{\MM}{{\cal M}}
 
 \newcommand{\PP}{{\cal P}}

 \newcommand{\dl}{\mathrm{L}}
 \newcommand{\nl}{\mathrm{NL}}

 \newcommand{\poly}{\mathrm{poly}}

\theoremstyle{plain}
\theoremheaderfont{\bfseries}
 \newtheorem{theorem}{Theorem}[section]
 \newtheorem{lemma}[theorem]{Lemma}
 \newtheorem{proposition}[theorem]{Proposition}
 \newtheorem{corollary}[theorem]{Corollary}

 {\theorembodyfont{\rmfamily}
  \newtheorem{definition}[theorem]{Definition}}
 {\theorembodyfont{\rmfamily} }
 {\theorembodyfont{\rmfamily} }

 \newtheorem{claim}{Claim}

 \newenvironment{proof}{\par \noindent
            {\bf Proof. \hs{2}}}{\hfill$\Box$ \vspace*{3mm}}

 \newenvironment{proofof}[1]{\vspace*{5mm} \par \noindent
         {\bf Proof of #1.\hs{2}}}{\hfill$\Box$ \vspace*{3mm}}

 \newcommand{\ceilings}[1]{\lceil #1 \rceil}
 
 \newcommand{\pair}[1]{\langle #1 \rangle}

\newcommand{\ignore}[1]{}

\newcommand{\cent}{|\!\! \mathrm{c}}
\newcommand{\dollar}{\$}

 \newcommand{\psublin}{\mathrm{PsubLIN}}
 \newcommand{\dstcon}{\mathrm{DSTCON}}
 \newcommand{\ptimespace}{\mathrm{PTIME,}\mathrm{SPACE}}
 \newcommand{\timespace}{\mathrm{TIME,}\mathrm{SPACE}}
 
 \newcommand{\unary}{\mathrm{unary}}

\begin{document}

\pagestyle{plain}
\setcounter{page}{1}

\begin{center}
{\Large {\bf State Complexity Characterizations of Parameterized Degree-Bounded Graph Connectivity, Sub-Linear Space Computation, and the Linear Space Hypothesis}}\footnote{This paper extends and corrects a preliminary report that has appeared in the Proceedings of the 20th International Conference on Descriptional Complexity of Formal Systems (DCFS  2018), Halifax, Canada, July 25--27, 2018. Lecture Notes in Computer Science, Springer, vol. 10952, pp. 237--249, 2018.}
\bs\s\\

{\sc Tomoyuki Yamakami}\footnote{Present Affiliation: Faculty of Engineering, University of Fukui, 3-9-1 Bunkyo, Fukui 910-8507, Japan}
\bs\\
\end{center}

\begin{quote}
\n{\bf Abstract.}
The linear space hypothesis is a practical working hypothesis,  which originally states the insolvability of a restricted 2CNF Boolean formula satisfiability problem parameterized by the number of Boolean variables. From this hypothesis, it naturally follows that the degree-3 directed graph connectivity problem (3DSTCON) parameterized by the number of vertices in a given graph cannot belong to PsubLIN, composed of all parameterized decision problems computable by polynomial-time, sub-linear-space deterministic Turing machines.
This hypothesis immediately implies L$\neq$NL and it was used as a solid foundation to obtain new lower bounds on the computational complexity of various NL search and NL optimization problems.
The state complexity of transformation refers to the cost of converting one type of finite automata to another type,
where the cost is measured in terms of the increase of the number of inner states of the converted automata from that of the original automata.
We relate the linear space hypothesis to the state complexity of transforming restricted 2-way nondeterministic finite automata to computationally equivalent 2-way alternating finite automata having narrow computation graphs.
For this purpose, we present state complexity characterizations of 3DSTCON and PsubLIN. We further characterize a nonuniform version of the linear space hypothesis in terms of the state complexity of transformation.

\s
{Keywords:} State Complexity, Alternating Finite Automata, Sub-Linear-Space Computability, Directed Graph Connectivity Problem, Parameterized Decision Problems, Polynomial-Size Advice, Linear Space Hypothesis
\end{quote}

\sloppy
\section{Prologue}\label{sec:introduction}

We provide the background of parameterized decision problems, the linear space hypothesis, and nonuniform state complexity. We then give an overview of major results of this work.

\subsection{Parameterized Problems
and the Linear Space Hypothesis}\label{sec:parameter}

The \emph{nondeterministic logarithmic-space} complexity class NL has been discussed since early days of computational complexity theory. Typical NL decision problems include the \emph{2CNF Boolean formula satisfiability problem} (2SAT) as well as the \emph{directed $s$-$t$ connectivity problem}\footnote{This problem is also known as the graph accessibility problem as well as the graph reachability problem.}  ($\dstcon$) of deciding the existence of a path from a vertex $s$ to another vertex $t$ in a given directed graph $G$. These problems are known to be  $\nl$-complete under log-space many-one reductions. The $\nl$-completeness is so  robust that even if we restrict our interest within graphs whose vertices are limited to be of degree at most $3$, the corresponding decision problem,  $3\dstcon$, remains $\nl$-complete.
Similarly, although we force 2CNF Boolean formulas in 2SAT to take only   variables, each of which appears at most 3 times in the form of literals, the obtained decision problem, $2\mathrm{SAT}_3$, is still an $\nl$-complete problem.

When we discuss the computational complexity of given problems, we in practice tend to be more concerned with various \emph{parameterizations} of the problems. We treat the size of specific ``input objects'' given to a problem as a ``practical'' \emph{size parameter} $n$ and use it to measure how much resource is needed for an algorithm to solve this problem.
There are, in fact, multiple ways to choose such a size parameter for each given problem. For example, given an instance $x=\pair{G,s,t}$ to $3\dstcon$, where $G$ is a directed graph (or a digraph) and $s,t$ are vertices, we often use as a size parameter the number of vertices in $G$ or the number of edges in $G$.
To emphasize the choice of a particular size parameter $m:\Sigma^*\to\nat$ for a decision problem $L$ over an alphabet $\Sigma$, it is convenient for us to express the problem as $(L,m)$, which gives rise to a \emph{parameterized decision problem}, where $\nat$ is the set of natural numbers.
Since we deal only with such parameterized problems in the rest of this paper, we occasionally drop the adjective ``parameterized'' as long as it is clear from the context.

Any instance $x=\pair{G,s,t}$ to $3\dstcon$ is usually parameterized by the numbers of vertices and of edges in the graph $G$.
It was shown in \cite{BBRS98} that $\dstcon$ with $n$ vertices and $m$ edges can be solved in $O(m+n)$ steps using only $n^{1-c/\sqrt{\log{n}}}$ space for a suitable constant $c>0$. However, it is unknown whether we can reduce this space usage down to $n^{\varepsilon}\, polylog(m+n)$ for a certain fixed constant $\varepsilon\in[0,1)$. Such a bound is informally called ``sub-linear'' \emph{in a strong sense}.
It has been conjectured that, for every constant $\varepsilon\in[0,1)$, no polynomial-time $O(n^{\varepsilon})$-space algorithm solves $\dstcon$ with $n$ vertices (see references in, e.g., \cite{ACL+14,CT15}).
For convenience, we denote by $\psublin$ the collection of all parameterized decision problems $(L,m)$ solvable deterministically in time polynomial in $|x|$ using space at most $m(x)^{\varepsilon}\ell(|x|)$ for certain constants $\varepsilon\in[0,1)$ and certain polylogarithmic (or polylog, for short) functions $\ell$ \cite{Yam17a}.

The \emph{linear space hypothesis} (LSH), proposed in \cite{Yam17a}, is a practical  working hypothesis, which originally asserts the insolvability of  $2\mathrm{SAT}_3$, together with the size parameter $m_{vbl}(\phi)$ indicating the number of variables in each given Boolean formula $\phi$, in polynomial time using sub-linear  space (namely, $(2\mathrm{SAT}_3,m_{vbl})\notin \psublin$).
As noted in \cite{Yam17a}, it is unlikely that $2\mathrm{SAT}$ replaces $2\mathrm{SAT}_3$ in the above definition of LSH. From this hypothesis, nevertheless, we immediately obtain the separation $\dl\neq\nl$, which many researchers believe to be true.
It was also shown in \cite{Yam17a} that, in the definition of LSH,  $(2\mathrm{SAT}_3,m_{vbl})$ can be replaced by $(3\dstcon,m_{ver})$, where $m_{ver}(\pair{G,s,t})$ refers to the number of vertices in $G$.
In the literature \cite{Yam17a,Yam17b}, LSH has acted as a reasonable foundation to obtain better lower bounds on the space complexity of several $\nl$-search and $\nl$-optimization problems.
To find more applications of this hypothesis, it is desirable to translate the hypothesis into other fields of interest.
In this work, we look for a statement, in automata theory, which is logically equivalent to LSH, in hope that we would find more useful applications of LSH in this theory.

\subsection{Families of Finite Automata and Families of Languages}\label{sec:families-languages}

The purpose of this work is to look for an automata-theoretical statement that is logically equivalent to the linear space hypothesis; in particular, we seek a new characterization of the relationship between $3\dstcon$ and $\psublin$ in terms of the state complexity of transforming a certain type of finite automata to another type with \emph{no direct reference} to $3\dstcon$ or $\psublin$.

It is often cited from \cite{BL77} (re-proven in \cite[Section 3]{Kap14}) that, if $\dl=\nl$, then every $n$-state \emph{two-way nondeterministic finite automaton}  (or 2nfa, for short) can be converted into an $n^{O(1)}$-state \emph{two-way deterministic finite automaton} (or 2dfa) that agrees with it on all inputs of length at most $n^{O(1)}$. Conventionally, we call by \emph{unary finite automata} automata working only on unary inputs (i.e, inputs over a one-letter alphabet).
Geffert and Pighizzini \cite{GP11} strengthened the aforementioned result by proving that the assumption of $\dl=\nl$ leads to the following: for any $n$-state unary 2nfa, there is a unary 2dfa of at most $n^{O(1)}$-states agreeing with it on all strings of length at most $n$.
Within a few years, Kapoutsis \cite{Kap14} gave a similar characterization using $\dl/\poly$, a nonuniform version of $\dl$, which states: $\nl\subseteq \dl/\poly$ if and only if (iff) there is a polynomial $p$ for which any $n$-state 2nfa has a 2dfa of at most $p(n)$ states agreeing with the 2nfa on strings of length at most $n$. Another incomparable characterization was given by Kapoutsis and Pighizzini \cite{KP15}: $\nl\subseteq\dl/\poly$ iff there is a polynomial $p$ satisfying that any $n$-state unary 2nfa has an equivalent unary 2dfa with a number of states at most $p(n)$.
It can be expected to find a similar automata characterization for LSH.

Earlier, Sakoda and Sipser \cite{SS78} laid out a complexity-theoretical framework to discuss state complexity of transformation by giving formal definitions to \emph{nonuniform state complexity classes} (such as $2\mathrm{D}$, $2\mathrm{N}/\poly$, $2\mathrm{N}/\unary$), each of which is composed of nonuniform families of ``promise decision problems'' (or ``partial decision problems'') recognized by finite automata of specified types and input sizes. Such complexity-theoretical treatments of families of finite automata were also considered by Kapoutsis \cite{Kap12,Kap14} and Kapoutsis and Pighizzini \cite{KP15} to establish relationships between nonuniform state complexity classes and nonuniform space-bounded complexity classes.
For those nonuniform state complexity classes, it was proven in \cite{Kap14,KP15} that $2\mathrm{N}/\poly \subseteq 2\mathrm{D}$ iff $\nl\subseteq \dl/\poly$ iff $2\mathrm{N}/\mathrm{unary}\subseteq 2\mathrm{D}$.

The first step of this work is our discovery of the fact that a family of promise decision problems is \emph{more closely related} to a parameterized decision problem than any standard decision problem (whose complexity is measured in terms of the binary encoding size of inputs).
Given a parameterized decision problem $(L,m)$, we naturally identify it with a family $\{(L_n,\overline{L}_n)\}_{n\in\nat}$ of promise decision problems defined by $L_n=\{x\in L \mid m(x)=n\}$ and $\overline{L}_n=\{x\in \overline{L} \mid m(x)=n\}$ for each index $n\in\nat$.
On the contrary, given a family $\{(A_n,B_n)\}_{n\in\nat}$ of promise decision problems over an alphabet $\Sigma$ satisfying  $\Sigma^*=\bigcup_{n\in\nat}(A_n\cup B_n)$, if we define $L = \bigcup_{n\in\nat}A_n$ and set $m(x)=n$ for each instance $x\in A_n\cup B_n$, then $(L,m)$ is a parameterized decision problem.
This identity eventually leads us to establish a new characterization of LSH, which will be discussed in Section \ref{sec:main-result}.

\subsection{Main Contributions}\label{sec:main-result}

As the main contribution of this work, firstly we provide two characterizations of $3\dstcon$ and $\psublin$ in terms of 2nfa's and  \emph{two-way alternating finite automata} (or 2afa's, for short), each of which alternatingly takes universal states and existential states, producing alternating $\forall$-levels and $\exists$-levels in its \emph{rooted computation tree} made up of (surface) configurations. Notice that 2nfa's are a special case of 2afa's. Secondly, we give a characterization of LSH in terms of the state complexity of transforming a restricted form of 2nfa to another restricted form of 2afa's.
The significance of our characterization includes the fact that LSH can be expressed completely by the state complexity of finite automata of certain types \emph{with no clear reference} to $(2\mathrm{SAT}_3,m_{vbl})$,   $(3\dstcon,m_{ver})$, or even $\psublin$; therefore, this characterization may help us apply LSH to a wider range of $\nl$-complete problems, which have little resemblance to $2\mathrm{SAT}_3$ and $3\dstcon$.

To handle an instance $(G,s,t)$ given to $3\dstcon$ on a Turing machine, we intend to use a ``reasonable'' encoding of $(G,s,t)$, where such an encoding must contain information on all vertices and edges of $G$ as well as the designated vertices $s$ and $t$ using $O(n\log{n})$ bits, where $n$ is the number of all vertices, and there must be an efficient way to retrieve all the information from this encoding.

To describe our result precisely, we further need to explain our terminology. A \emph{simple 2nfa} is a 2nfa having a ``circular'' input tape\footnote{A 2nfa  with a tape head that sweeps a circular tape is called ``rotating'' in   \cite{Kap09,KP15}.}  (in which both endmarkers are located next to each other) whose tape head ``sweeps'' the tape (i.e. it moves only to the right), and making nondeterministic choices only at the right endmarker.\footnote{This requirement is known in \cite{GGP14,KP15} as ``outer nondeterminism.'' In Section \ref{sec:model-two-way}, we will call it ``end-branching'' for both 2nfa's and 2afa's.}
For a positive integer $c$, a \emph{$c$-branching 2nfa} makes only at most $c$ nondeterministic choices at every step and a family of 2nfa's is called \emph{$c$-branching} if every 2nfa in this family is $c$-branching.
A computation of an 2afa is normally expressed as a tree whose nodes are labeled by (surface) configurations; however, it is more convenient and more concise to view the computation as a ``graph,'' in which the same configurations at the same depth are all treated as the same vertex. Such a graph is particularly called a \emph{computation graph}. A \emph{$c$-narrow 2afa} is a 2afa whose computation graphs have width (i.e., the total number of distinct vertices at a given level)
bounded by $c$ at every $\forall$-level. Notice that a $c$-branching 2nfa is, in general, not $c$-narrow.

For convenience, we say that a finite automaton $M_1$ is \emph{(computationally) equivalent} to another finite automaton $M_2$ over the same input alphabet if $M_1$ agrees with $M_2$ on all inputs. Here, we use a straightforward binary encoding $\pair{M}$ of an $n$-state finite automaton $M$ using  $O(n\log{n})$ bits.
A family $\{M_n\}_{n\in\nat}$ is said to be \emph{$\dl$-uniform} if a deterministic Turing machine (or a DTM) produces from $1^n$ an encoding $\pair{M_n}$ of finite automaton $M_n$ on its write-only output tape using space logarithmic in $n$.

\begin{theorem}\label{3DSTCON-char-uniform}
The following three statements are logically equivalent.
\renewcommand{\labelitemi}{$\circ$}
\begin{enumerate}\vs{-2}
  \setlength{\topsep}{-2mm}%
  \setlength{\itemsep}{1mm}%
  \setlength{\parskip}{0cm}%

\item The linear space hypothesis (LSH) fails.

\item For any constant $c>0$, there exists a constant  $\varepsilon\in[0,1)$ such that, for any constant $e>0$, every $\dl$-uniform family of $c$-branching  simple 2nfa's with at most $en\log{n}+e$ states can be converted into another $\dl$-uniform family of equivalent $O(n^{\varepsilon})$-narrow 2afa's with $n^{O(1)}$ states.

\item For any constant $c>0$, there exists a constant $\varepsilon\in[0,1)$ and a log-space computable function that, on every input of an encoding of a $c$-branching simple $n$-state 2nfa,
    produces another encoding of an equivalent $O(n^{\varepsilon})$-narrow 2afa with $n^{O(1)}$ states.
\end{enumerate}\vs{-2}
Furthermore, even if we fix $c$ to $3$, the above three statements are logically equivalent.
\end{theorem}

Our proof of Theorem \ref{3DSTCON-char-uniform} is based on two explicit characterizations, given in Section \ref{sec:basic-characterize}, of $\psublin$ and $(3\dstcon,m_{ver})$ in terms of state complexities of restricted 2nfa's and of restricted 2afa's, respectively. Concerning Statement (2) of Theorem \ref{3DSTCON-char-uniform}, it seems difficult to construct equivalent  $O(n^{\varepsilon})$-narrow 2afa's from any given simple 2nfa's but it is possible to achieve the $O(n^{1-c/\sqrt{\log{n}}})$-narrowness of 2afa's
for a certain constant $c>0$.

\begin{proposition}\label{Barnes-translate}
Every $\dl$-uniform family of constant-branching $O(n\log{n})$-state simple 2nfa's can be converted into another $\dl$-uniform family of equivalent $O(n^{1-c/\sqrt{\log{n}}})$-narrow 2afa's with $n^{O(1)}$ states for a certain constant $c>0$.
\end{proposition}

In addition to the original linear space hypothesis, it is possible to discuss its  \emph{nonuniform version}, which asserts that $(2\mathrm{SAT}_3,m_{ver})$ does not belong to  a nonuniform version of $\psublin$, succinctly denoted by $\psublin/\poly$. The nonuniform setting can provide a more concise characterization of LSH than Theorem \ref{3DSTCON-char-uniform} does.

The nonuniform state complexity class $2\mathrm{qlinN}$ consists of all families $\{(L_n,\overline{L}_n)\}_{n\in\nat}$ of promise decision problems, each $(L_n,\overline{L}_n)$ of which is recognized by a certain $c$-branching simple $O(n\log^{k}{n})$-state 2nfa on all inputs for appropriate constants $c,k\in\nat-\{0\}$ (where ``qlin'' indicates ``quasi-linear'').
Moreover, $2\mathrm{A}_{narrow(f(n))}$ is composed of families $\{(L_n,\overline{L}_n)\}_{n\in\nat}$ of promise decision problems recognized by $O(f(n))$-narrow 2afa's using at most $p(n)$ states.

\begin{theorem}\label{3DSTCON-char-nonunif}
The following three statements are logically equivalent.
\renewcommand{\labelitemi}{$\circ$}
\begin{enumerate}\vs{-2}
  \setlength{\topsep}{-2mm}%
  \setlength{\itemsep}{1mm}%
  \setlength{\parskip}{0cm}%

\item The nonuniform linear space hypothesis fails.

\item For any constant $c>0$, there exists a constant $\varepsilon\in[0,1)$ such that every $c$-branching simple $n$-state 2nfa can be converted into  an equivalent $O(n^{\varepsilon})$-narrow 2afa with at most $n^{O(1)}$ states.

\item $2\mathrm{qlinN} \subseteq \bigcup_{\varepsilon\in[0,1)} 2\mathrm{A}_{narrow(n^{\varepsilon})}$.
\end{enumerate}
\end{theorem}

So far, we work mostly on input alphabets of size at least $2$. In contrast, if we turn our attention to ``unary'' alphabets and finite automata over such unary alphabets (which are succinctly called \emph{unary (finite) automata}), then we obtain only a slightly weaker implication to the failure of LSH.

\begin{theorem}\label{uniform-unary-char}
Each of the following statements implies the failure of the linear space hypothesis.
\renewcommand{\labelitemi}{$\circ$}
\begin{enumerate}\vs{-2}
  \setlength{\topsep}{-2mm}%
  \setlength{\itemsep}{1mm}%
  \setlength{\parskip}{0cm}%

\item For any constant $c>0$, there exists a constant $\varepsilon\in[0,1)$ such that, for any constant $e>0$, every $\dl$-uniform family of $c$-branching simple unary 2nfa's with at most $en^3\log{n}+e$ states can be converted into an $\dl$-uniform family of equivalent $O(n^{\varepsilon})$-narrow simple unary 2afa's with $n^{O(1)}$ states.

\item For any constant $c>0$, there exist a constant $\varepsilon\in[0,1)$ and a log-space computable function $f$ such that, for any constant $e>0$, on every input of an encoding of $c$-branching simple unary 2nfa with at most $en^3\log{n}+e$ states, $f$ produces another encoding of equivalent  $O(n^{\varepsilon})$-narrow simple unary 2afa having $n^{O(1)}$ states.
\end{enumerate}\vs{-2}
Furthermore, it is possible to fix $c$ to $3$ in the above statements.
\end{theorem}

Theorems \ref{3DSTCON-char-uniform}--\ref{3DSTCON-char-nonunif} will be proven in Section \ref{sec:main-proof} after we establish basic properties of $\psublin$ and $3\dstcon$ in Section \ref{sec:basic-characterize}. Theorem  \ref{uniform-unary-char} will be verified in Section \ref{sec:unary-case}.

\section{Supporting Terminology}\label{sec:preliminaries}

We have briefly discussed in Section \ref{sec:introduction} key terminology necessary to state our main contributions. Here, we further explain their supporting terminology.

\subsection{Numbers, Languages, and Size Parameters}\label{sec:numbers}

We denote by  $\nat$  the set of all \emph{natural numbers} (i.e., nonnegative integers) and set $\nat^{+}=\nat-\{0\}$. For two integers $m$ and $n$ with $m\leq n$, an \emph{integer interval} $[m,n]_{\integer}$ is a set $\{m,m+1,m+2,\ldots,n\}$. Given a set $A$, $\PP(A)$ expresses the \emph{power set} of $A$; that is, the set of all subsets of $A$. We assume that all \emph{polynomials} have nonnegative integer coefficients and all \emph{logarithms} are to the  base $2$.
A function $f:\nat\to\nat$ is \emph{polynomially bounded} if there exists a polynomial $p$ satisfying $f(n)\leq p(n)$ for all $n\in\nat$.

An \emph{alphabet} is a finite nonempty set of ``symbols'' or ``letters,'' and a \emph{string} over such an alphabet $\Sigma$ is a finite sequence of symbols in $\Sigma$. The \emph{length} of a string $s$, denoted by $|s|$, is the total number of symbols in $s$. We use $\lambda$ to express the \emph{empty string} of length $0$. Given an alphabet $\Sigma$, the notation $\Sigma^{\leq n}$ (resp., $\Sigma^n$) indicates the set of all strings of length at most (resp., exactly) $n$ over $\Sigma$. We write $\Sigma^*$ for $\bigcup_{i\geq0}\Sigma^{i}$.
A \emph{language} over $\Sigma$ is a set of strings over $\Sigma$. The \emph{complement} of $L$ is $\Sigma^*-L$ and is succinctly denoted by $\overline{L}$ as long as $\Sigma$ is clear from the context.
For convenience, we abuse the notation $L$ to indicate its \emph{characteristic function} as well; that is, $L(x)=1$ for all $x\in L$, and $L(x)=0$ for all $x\in\overline{L}$. A function $f:\Sigma^*\to\Gamma^*$ for two alphabets $\Sigma$ and $\Gamma$ is \emph{polynomially bounded} if there is a polynomial $p$ such that $|f(x)|\leq p(|x|)$ for all $x\in\Sigma^*$.

A \emph{size parameter} is a function mapping $\Sigma^*$ to $\nat^{+}$, which is used as a base unit in our analysis. We call a size parameter $m$ \emph{ideal} if there are constants $c_1,c_2>0$ and $k\geq1$ such that $c_1 m(x)\leq |x|\leq c_2m(x)\log^k{m(x)}$ for all $x$ with $|x|\geq2$.

Given a number $i\in\nat$, $binary(i)$ denotes the \emph{binary representation} of $i$; for example, $binary(0)=0$, $binary(1)=1$, and $binary(5)=101$.
For a length-$n$ string $x=x_1x_2\cdots x_n$ over the alphabet $\{0,1,\#,\bot\}$ (where $\#,\bot\notin\{0,1\}$), we define its \emph{binary encoding} $\pair{x}_2$ as $\widehat{x_1}\widehat{x_2}\cdots \widehat{x_n}$, where $\widehat{0}=00$, $\widehat{1}=01$, $\widehat{\#}=11$, and $\widehat{\bot}=10$. For example, $\pair{10}_2=0100$ and $\pair{0\#1\bot1}_2 = 0011011001$. We also use a \emph{fixed-length binary representation} $bin_n(i)$ defined to be  ${0^{k-1}1}binary(i)$ with $n=k+|binary(i)|$ and $k\geq1$. For example, we obtain $bin_5(2)=00110$,  $bin_4(5)=1101$, and $bin_3(1)=011$.

A \emph{promise decision problem} over an alphabet $\Sigma$ is a pair $(L_1,L_2)$ of disjoint subsets of $\Sigma^*$. We interpret $L_1$ and $L_2$ into sets of \emph{accepted} (or \emph{positive}) \emph{instances} and of \emph{rejected} (or \emph{negative}) \emph{instances}, respectively. In this work, we consider a family $\{(L_n,\overline{L}_n)\}_{n\in\nat}$ of promise decision problems \emph{over the same alphabet}.\footnote{Some of the literature  consider families $\LL$ of promise decision problems, each $(L_n,\overline{L}_n)$ of which may use a different alphabet. Nonetheless, for the purpose of this work, we do not take such an approach toward families of promise decision problems.}

As noted in Section \ref{sec:families-languages}, there is a direct translation  between parameterized decision problems and families of promise decision problems.
Given a parameterized decision problem $(L,m)$ over an alphabet $\Sigma$, a family $\LL=\{(L_n,\overline{L}_n)\}_{n\in\nat}$ is said to be \emph{induced from} $(L,m)$ if, for each index $n\in\nat$, $L_n=L\cap\Sigma_n$ and $\overline{L}_n=\overline{L}\cap \Sigma_n$, where $\Sigma_n=\{x\in\Sigma^*\mid m(x)=n\}$. As a refinement, we also set $L_{n,l}= L_n\cap \Sigma^l$ and $\overline{L}_{n,l} = \overline{L}_n\cap \Sigma^l$
for every pair $n,l\in\nat$.
In the rest of this paper, for our convenience, we identify parameterized problems with their induced families of promise problems.

\subsection{Turing Machine Models}\label{sec:Turing-machine}

To discuss space-bounded computation, we consider only the following form of 2-tape Turing machines.
A \emph{nondeterministic Turing machine} (or an NTM, for short) $M$ is a tuple $(Q,\Sigma,\{\cent,\dollar\},\Gamma,\delta,q_0,Q_{acc},Q_{rej})$ with a read-only input tape (over the \emph{extended alphabet}  $\check{\Sigma}=\Sigma\cup\{\cent,\dollar\}$) and a rewritable work tape (over the tape alphabet $\Gamma$). The transition function $\delta$ of $M$ is a map from $(Q-Q_{halt})\times \check{\Sigma}\times \Gamma$ to $\PP(Q\times\Gamma\times D_1\times D_2)$ with $D_1=D_2=\{+1,-1\}$ and $Q_{halt}=Q_{acc}\cup Q_{rej}$.
We assume that $\Gamma$ contains a distinguished blank symbol $B$.
In contrast, when $\delta$ maps to $Q\times\Gamma\times D_1\times D_2$, $M$ is called a \emph{deterministic Turing machine} (or a DTM) and can be treated as a special case of NTMs.
Each input string $x$ is written between $\cent$ (left endmarker) and $\dollar$ (right endmarker) on the input tape and all cells in the input tape are indexed from the left to the right by the integers incrementally from $0$ to $|x|+1$, where $\cent$ and $\dollar$ are respectively placed at cell $0$ and cell $|x|+1$. The work tape is a semi-infinite tape, stretching to the right, the leftmost cell of the tape is indexed $0$, and all other cells are consecutively numbered to the right as $1,2,\ldots$.

Given a Turing machine $M$ and an input $x$, a \emph{surface  configuration} is a quadruple $(q,j,k,w)$, where $q\in Q$, $j\in[0,|x|+1]_{\integer}$, $k\in\nat$, and $w\in\Gamma^*$, which represents the circumstance in which $M$ is in state $q$, scanning a symbol at cell $j$ of the input tape and a symbol at cell $k$ of the work tape containing string $w$.
For each input $x$, an NTM $M$ produces a \emph{computation tree}, in which each node is labeled by a surface configuration of $M$ on $x$ and an edge from a parent node to its children indicates a single transition of $M$ on $x$.
An NTM \emph{accepts} $x$ if it starts with the state $q_0$ scanning $\cent$ and, along a certain path of the computation tree, it enters an accepting state and halts. Such a computation path is called an \emph{accepting path}; in contrast, a \emph{rejection path} is a computation path that terminates in a rejecting state. If all computation paths are either rejecting paths or non-terminating paths, then we simply say that $M$ \emph{rejects} $x$.

In this work, we generally use Turing machines to solve parameterized decision problems. Occasionally, we also use Turing machines to compute functions. For this purpose, we need to append an extra \emph{semi-infinite write-only output tape} to each DTM, where a tape is \emph{write only} if its tape head must move to the right whenever it write a non-blank symbol onto the tape.

A function $f:\nat\to\nat$ is \emph{log-space computable} if there exists a DTM such that, for each given length $n\in\nat$,  $M$ takes $1^n$ as its input and then produces $1^{f(n)}$ on a write-only output tape using $O(\log{n})$ work space. In contrast, a function $m:\Sigma^*\to\nat^{+}$ is called a \emph{log-space size parameter} if there exists a DTM $M$ that, on any input $x$, produces $1^{m(x)}$ (\emph{in unary}) on its output tape using only $O(\log{n})$ work space \cite{Yam17a}. This implies that $m(x)$ is upper-bounded by $|x|^{O(1)}$. Concerning space constructibility, we here take the following time-bounded version.
A function $s:\nat\to\nat$ is \emph{$t(n)$-time space constructible} if there exists a DTM $M$ such that, for each given length $n\in\nat$, when $M$ takes $1^n$ as an input written on the input tape and $M$ produces $1^{s(n)}$ on its output tape and halts within $t(n)$ steps using no more than $s(n)$ cells. In a similar vein, a function $t:\nat\to\nat$ is \emph{log-space time constructible} if there is a DTM $M$ such that, for any $n\in\nat$, $M$ starts with $1^n$ as an input and halts exactly in $t(n)$ steps using $O(\log{n})$ space.

\subsection{Sub-Linear-Space Computability and Advice}\label{sec:sub-linear-space}

Take two functions $s:\nat\times\nat\to\nat^{+}$ and $t:\nat\to\nat^{+}$, and let $m$ denote any size parameter. The notation $\timespace(t(x),s(x,m(x)))$ (where $x$ expresses merely a \emph{symbolic input}) denotes the collection of all parameterized decision problems $(L,m)$ recognized by DTMs (each of which is equipped with a read-only input tape and a semi-infinite rewritable work tape)  within time $c_1t(x)+c_1$ using space at most $c_2s(x,m(x))+c_2$ on every input $x$ for certain absolute constants $c_1,c_2>0$.
The parameterized complexity class $\ptimespace(s(x,m(x)))$, defined in \cite{Yam17a}, is the union of all classes $\timespace(p(|x|),s(x,m(x)))$ for any positive polynomial $p$.

Karp and Lipton \cite{KL81} supplemented extra information, represented by \emph{advice strings}, to underlying Turing machines to enhance the computational power of the machines. More formally, we first equip our underlying machine with an additional read-only \emph{advice tape}, to which we provide exactly one advice string, surrounded by two endmarkers ($\cent$ and $\dollar$), of pre-determined length for all instances of each fixed length.

Let $h$ be an arbitrary function from $\nat$ to $\nat$. The nonuniform complexity class $\timespace(t(x),s(x,m(x)))/h(|x|)$ is obtained from $\timespace(t(x),s(x,m(x)))$
by providing underlying Turing machines with an advice string of length $h(n)$ for all instances of each length $n$. For $\ptimespace(s(x,m(x)))$, its advised variant $\ptimespace(s(x,m(x)))/\poly$ can be defined similarly by supplementing advice of polynomial size.

\begin{definition}
The class $\psublin$ is defined to be the union of all $\ptimespace(m(x)^{\varepsilon}\ell(|x|))$ for any log-space size parameter $m$, any  constants $k\geq1$ and $\varepsilon\in[0,1)$, and any polylog function  $\ell$. Similarly, we define $\psublin/\poly$ as an advised  version of $\psublin$ using $\ptimespace(m(x)^{\varepsilon}\ell(|x|))/\poly$.
\end{definition}

\subsection{Models of Two-Way Finite Automata}\label{sec:model-two-way}

We consider two-way finite automata, equipped with a read-only input tape and a tape head that moves along this input tape in both directions (to the left and to the right) bit never stays still at any tape cell. To clarify the use of two endmarkers $\cent$ and $\dollar$, we explicitly include them in the description of finite automata.

Let us start with defining \emph{two-way nondeterministic finite automata} (or 2nfa's).  A 2nfa $M$ is formally a septuple  $(Q,\Sigma,\{\cent,\dollar\},\delta,q_0,Q_{acc},Q_{rej})$, where $Q$ is a finite set of inner states, $\Sigma$ is an input alphabet with $\check{\Sigma}=\Sigma\cup\{\cent,\dollar\}$, $q_0$ ($\in Q)$ is the initial state, $Q_{acc}$ and $Q_{rej}$ are respectively sets of accepting and rejecting states satisfying both $Q_{acc}\cup Q_{rej}\subseteq Q$ and $Q_{acc}\cap Q_{rej}=\setempty$, and $\delta$ is a transition function from $(Q-Q_{halt})\times \check{\Sigma}$ to $\PP(Q\times D)$ with $D=\{+1,-1\}$ and  $Q_{halt} = Q_{acc}\cup Q_{rej}$. We always assume that $\cent,\dollar\notin \Sigma$.
The 2nfa $M$ behaves as follows. Assuming that $M$ is in state $q$ scanning symbol $\sigma$, if a transition has the form $(p,d)\in\delta(q,\sigma)$, then $M$ changes its inner state to $p$, moves its tape head in direction $d$ (where $d=+1$ means ``to the right'' and $d=-1$ means ``to the left.'').

A \emph{two-way deterministic finite automaton} (or a 2dfa) is defined
as a septuple $(Q,\Sigma,\{\cent,\dollar\},\delta,q_0,Q_{acc},Q_{rej})$, which is similar to a 2nfa but its transition function $\delta$ is a map from $(Q-Q_{halt})\times \check{\Sigma}$ to $Q\times D$. As customarily, we view 2dfa's as a special case of 2nfa's by identifying a singleton $\{(q,d)\}$ with its element $(q,d)$.

An input tape is called \emph{circular} if the right of cell $\dollar$ is cell $\cent$ and the left of cell $\cent$ is cell $\dollar$. Hence, on this circular tape, when a tape head moves off the right of $\dollar$ (resp., the left of $\cent$), it instantly reaches $\cent$ (resp., $\dollar$). A circular-tape finite automaton is \emph{sweeping} if the tape head always moves to the right.
A circular-tape 2nfa is said to be \emph{end-branching} if it makes nondeterministic choices only at the cell containing $\dollar$. A \emph{simple 2nfa} is a 2nfa that has a circular input tape, is sweeping, and is end-branching.

For a fixed integer $c\geq1$, a 2nfa $M$ is said to be \emph{$c$-branching} if, for any inner state $q$ and tape symbol $\sigma$, there are at most $c$ next moves (i.e., $|\delta(q,\sigma)|\leq c$).
Note that all 2dfa's are $1$-branching. We say that a family of 2nfa's is \emph{constant-branching} if every 2nfa in the family is $c$-branching for an absolute constant $c\geq1$.

Different from the case of Turing machine, a \emph{surface configuration} of a finite automaton $M$ is a pair $(q,i)$ with $q\in Q$ and $i\in\nat$. Since, for each input size $n$, $i$ ranges only over the integer interval $[0,n+1]_{\integer}$, the total number of surface configurations of an $n$-state finite automaton working on inputs of length $m$ is $n(m+2)$. Moreover, the notion of \emph{computation trees} and \emph{computation paths} can be naturally introduced as done for Turing machine in Section \ref{sec:Turing-machine}.

We use the following acceptance criteria: $M$ \emph{accepts} input $x$ if there is a finite accepting computation path of $M$ on $x$; otherwise, $M$ is said to \emph{reject} $x$. We say that $M$ \emph{accepts in time $t(n)$} if, for any length $n\in\nat^{+}$ and any input $x$ of length $n$, if $M$ accepts $x$, then there exists an accepting computation path of length at most $t(n)$. Let $L(M)$ express the set of all strings accepted by $M$. For any 2nfa $M$, if there is an accepting computation path of $M$ on input $x$, then its minimal length is at most $|Q|(|x|+2)$.

To characterize polynomial-time sub-linear-space computation, we further look into a model of \emph{two-way alternating finite automata} (or 2afa's) whose computation trees are particularly ``narrow.''
Formally, a 2afa is expressed as a tuple $(Q,\Sigma,\{\cent,\dollar\},\delta, q_0,Q_{\forall},Q_{\exists},Q_{acc},Q_{rej})$, where $Q$ is partitioned into a set $Q_{\forall}$ of \emph{universal states} (or $\forall$-states) and a set $Q_{\exists}$ of \emph{existential states} (or $\exists$-states). On each input, similar to a 2nfa, a 2afa $M$ branches out according to the value $\delta(q,\sigma)$ after scanning symbol $\sigma\in\check{\Sigma}$ in state $q\in Q$, and $M$ generates a \emph{computation tree} whose nodes are labeled by  surface configurations of $M$.
A \emph{$\{\forall,\exists\}$-label} of a node is defined as follows.
A node has a \emph{$\forall$-label} (resp., an \emph{$\exists$-label}) if its associated surface configuration has a universal state (resp., an existential state). A computation tree of $M$ on an input is said to be \emph{$\{\forall,\exists\}$-leveled} if all nodes of the same depth from the root node have the same label (either $\forall$ or $\exists$). When $|\delta(q,\sigma)|=1$, we customarily call this transition a \emph{deterministic transition} (or \emph{deterministic move}). Since all 2nfa's are 2afa's, we naturally extend the term ``end-branching'' to 2afa's by demanding that a 2afa always makes deterministic moves while reading symbols in $\Sigma\cup\{\cent\}$ and it branches out only at reading $\dollar$. Similarly to 2nfa's, we say that a 2afa is \emph{simple} if it has a circular input tape, is sweeping, and is end-branching.

Given an input $x$, a 2afa $M$ \emph{accepts} $x$ if there is an \emph{accepting computation subtree} $T$ of $M$ on $x$, in which (i) $T$ is rooted at the initial surface configuration, (ii)
$T$ contains exactly one branch from every node labeled by an $\exists$-state, (iii) $T$ contains all branches from each node having  $\forall$-labels, and (iv) all leaves of $T$ must have accepting states. Otherwise, we say that $M$ \emph{rejects} $x$.

Abusing the terminology, we say that a family $\MM=\{M_n\}_{n\in\nat}$ of 2afa's \emph{runs in $t(n,|x|)$ time} (where $x$ expresses a ``symbolic'' input) if, for any index $n\in\nat$ and for any input $x$ accepted by $M_n$, the height of a certain accepting computation subtree of $M_n$ on $x$ is bounded from above by $t(n,|x|)$. Since any computation path of an accepting computation subtree of $M$ on $x$ cannot have two identical surface configurations, the length of such a computation path must be at most $|Q|(|x|+2)$; therefore, the height of the accepting computation subtree is at most $|Q|(|x|+2)$. In other words, for any function $s:\nat\to\nat$, if $M_n$ has $s(n)$ states, then $\MM$ runs in $s(n)\cdot O(|x|)$ time on inputs $x$.

As noted in Section \ref{sec:main-result}, we can transform any computation tree into its associated \emph{computation graph} by merging all nodes of the computation tree at each level.
Let $f$ be any function on $\nat$. An \emph{$f(n)$-narrow 2afa} is a 2afa that, on each input $x$, produces a $\{\forall,\exists\}$-leveled computation graph that has width at most $f(|x|)$ at every $\forall$-level.

In general, we say that two machines $M_1$ and $M_2$ are \emph{(recognition) equivalent} if $M_1$ agrees with $M_2$ on all inputs.

\subsection{Nonuniform State Complexity and State Complexity Classes}\label{sec:state-complexity-class}

Given a finite automaton $M$, the \emph{state complexity} of $M$ refers to the number of $M$'s inner states.
The \emph{state complexity of transforming 2nfa's to 2dfa's} refers to the minimal cost of converting any 2nfa $M$ into a certain 2dfa $N$.  More precisely, when any given $n$-state 2nfa $M$ can be transformed into its equivalent $s(n)$-state 2dfa $N$, if $s(n)$ is the minimal number, then $s(n)$ is the state complexity of transforming 2nfa's to 2dfa's.

Instead of considering each finite automaton separately, here, we are concerned with a family or a collection $\{M_{n,l}\}_{n,l\in\nat}$ of finite automata, each $M_{n,l}$ of which has a certain number of states, depending only on parameterized sizes $n$ and input lengths $l$.
We say that a family $\{M_n\}_{n\in\nat}$ of machines \emph{solves} (or \emph{recognizes}) a family $\{(L_n,\overline{L}_n)\}_{n\in\nat}$ of promise decision problems if, for every $n\in\nat$, (1) for any $x\in L_n$, $M_n$ accepts $x$ and (2) for any $x\in\overline{L}_n$, $M_n$ rejects $x$. Notice that there is no requirement for any string outside of $L_n\cup \overline{L}_n$.
Kapoutsis \cite{Kap14} and Kapoutsis and Pighizzini \cite{KP15} presented a  characterization of $\nl\subseteq \dl/\poly$ in terms of the nonuniform
state complexity classes $2\mathrm{D}$, $2\mathrm{N}/\poly$, and $2\mathrm{N}/\mathrm{unary}$, which are described as collections of promise decision problems (or partial decision problems).

We formally define two critical classes $\mathrm{2qlinN}$ and $2\mathrm{A}_{narrow(f(n))}$ .

\begin{definition}
\renewcommand{\labelitemi}{$\circ$}
\begin{enumerate}
  \setlength{\topsep}{-2mm}%
  \setlength{\itemsep}{1mm}%
  \setlength{\parskip}{0cm}%

\item The nonuniform state complexity class $2\mathrm{qlinN}$ is the collection of nonuniform families $\{(L_n,\overline{L}_n)\}_{n\in\nat}$ of promise decision problems satisfying the following: there exist four constants  $c,k\in\nat^{+}$, $d,e>0$, and a nonuniform family $\{M_n\}_{n\in\nat}$ of $c$-branching 2nfa's such that, for each index $n\in\nat$, $M_n$ has at most $dn\log^k{n}+e$ states and $M_n$ solves  $(L_n,\overline{L}_n)$.

\item Given a function $f$ on $\nat$, we define $2\mathrm{A}_{narrow(f(n))}$ to be the collection of nonuniform families $\{(L_n,\overline{L}_n)\}_{n\in\nat}$ of promise decision problems, each $(L_n,\overline{L}_n)$ of which is solved by a certain  $n^{O(1)}$-state 2afa $N_n$ whose computation graphs are $O(f(n))$-narrow.
\end{enumerate}
\end{definition}

\section{Two Fundamental Automata Characterizations}\label{sec:basic-characterize}

Since Theorems \ref{3DSTCON-char-uniform}--\ref{3DSTCON-char-nonunif} are concerned with the language $3\dstcon$ and the parameterized complexity class $\psublin$, before proving these theorems, we intend to look into their basic properties in depth. In what follows, we will present two automata characterizations of $3\dstcon$ and $\psublin$.

\subsection{Automata Characterizations of PsubLIN}

Let us present a precise characterization of $\psublin$ in terms of narrow 2afa's. Interestingly, the narrowness of 2afa's directly corresponds to the space usage of DTMs. What we intend to prove in Proposition \ref{ptimespace-character} is, in fact, far more general than what we actually need for proving Theorems \ref{3DSTCON-char-uniform}--\ref{3DSTCON-char-nonunif}. We expect that such a general characterization could find other useful applications as well.

Firstly, let us recall the parameterized complexity class $\timespace(t(|x|),\ell(m(x)))$ from Section \ref{sec:sub-linear-space}. Our proof of Proposition \ref{ptimespace-character} requires a fine-grained analysis of the well-known transformation of \emph{alternating Turing machines} (or ATMs) to DTMs and vice versa.
In what follows, we freely identify a language with its \emph{characteristic function}.

\begin{proposition}\label{ptimespace-character}
Let $t:\nat\to\nat^{+}$ be log-space time constructible and let $\ell:\nat\to\nat^{+}$ be $t(n)$-time space constructible. Consider a language $L$ and a log-space size parameter $m$.
\renewcommand{\labelitemi}{$\circ$}
\begin{enumerate}\vs{-2}
  \setlength{\topsep}{-2mm}%
  \setlength{\itemsep}{1mm}%
  \setlength{\parskip}{0cm}%

\item If $(L,m)\in \timespace(t(|x|),\ell(m(x)))$, then there are three constants $c_1,c_2,c_3>0$ and an $\dl$-uniform family $\{M_{n,l}\}_{n,l\in\nat}$ of $c_2\ell(n)$-narrow 2afa's having at most $c_1t(l)\ell(n)$ states such that each $M_{m(x),|x|}$ computes $L(x)$ within $c_3t(|x|)$ time on all inputs $x$ given to $L$.

\item If  there are constants $c_1,c_2,c_3>0$ and an $\dl$-uniform family $\{M_{n,l}\}_{n,l\in\nat}$ of $c_2\ell(n)$-narrow 2afa's having at most $c_1t(l)\ell(n)$ states such that each $M_{m(x),|x|}$ computes $L(x)$ within $c_3t(|x|)$ time on all inputs $x$ to $L$, then $(L,m)$ belongs to $\timespace(|x|^{O(1)}+(t(|x|)\ell(m(x))|x|)^2, \ell(m(x))\log{|x|}+\log{t(|x|)})$.
\end{enumerate}
\end{proposition}

Hereafter, we will proceed to the proof of Proposition \ref{ptimespace-character}. Our proof is different from a well-known proof in \cite{CKS81}, which shows a simulation between space-bounded DTMs and time-bounded ATMs.
As an example of such difference, a simulation of an ATM by an equivalent space-bounded DTM  in \cite{CKS81} uses the \emph{depth-first traversal} of a computation tree, whereas we use the \emph{breadth-first traversal} because of the narrowness of 2afa's.

For our proof, moreover, we need an \emph{encoding scheme} of 2nfa's. For this purpose, we use the following binary encoding scheme for a $c$-branching 2-way finite automaton $M=(Q,\Sigma,\{\cent,\dollar\},\delta,q_0,Q_{acc},Q_{rej})$.
Firstly, we fix a linear order on $Q\times\check{\Sigma}$ as $\{e_1,e_2,\ldots,e_k\}$ with $k=|Q\times\check{\Sigma}|$ and assume an appropriate binary encoding $\pair{p,\sigma}$ for each pair $(p,\sigma)\in Q\times\check{\Sigma}$. The transition function $\delta$ can be viewed as a  $k\times c$ matrix, in which each row is indexed by a pair $(p,\sigma)\in Q\times\check{\Sigma}$ and it contains $c$ entries $(p_1,d_1),(p_2,p_2),\ldots,(p_c,d_c)$ in order, if  $\delta(p,\sigma)=\{(p_1,d_1),(p_2,d_2),\ldots,(p_l,d_l)\}$ for a certain value $l\in[0,c]_{\integer}$ and the others $(p_{l+1},d_{l+1}),\ldots,(p_{c},d_{c})$ must be a designated symbol $\bot$. Given such a pair $(p,\sigma)$, we set $C_{p,\sigma}$ as $\#\pair{p_1,d_1}\# \pair{p_2,d_2}\#\cdots \#\pair{p_c,\sigma_c}\#$.
Finally, we define an encoding of $M$, denoted by $\pair{M}$, to be $\pair{C_{e_1}\#^2 C_{e_2}\#^2 \cdots \#^2C_{e_k}}_2$.
It then follows that there exists a constant $e>0$ satisfying $|\pair{M}|\leq en\log{n}$ for any $c$-branching 2nfa $M$ having at most $n$ states.

\begin{proofof}{Proposition \ref{ptimespace-character}}
Meanwhile, for convenience, we allow 2afa's to make \emph{stationary moves}  (i.e., a tape head stays still at certain steps). Take a parameterized decision problem $(L,m)$ with a log-space size parameter $m$, a log-space time-constructible function $t:\nat\to\nat^{+}$, and an $t(n)$-time space-constructible function $\ell:\nat\to\nat^{+}$. Consider a family $\LL = \{(L_{n,l},\overline{L}_{n,l})\}_{n,l\in\nat}$ of promise decision problems such that $L_{n,l}=\{x\in L\cap \Sigma^l\mid m(x)=n\}$ and $\overline{L}_{n,l}=\{x\in \overline{L}\cap\Sigma^l\mid m(x)=n\}$ for any pair $n,l\in\nat$.

(1) Assume that $(L,m)$ belongs to $\timespace(t(|x|),\ell(m(x)))$. Let us consider a DTM $N = (Q,\Sigma,\{\cent,\dollar\},\Gamma,\delta,q_0,Q_{acc},Q_{rej})$ that solves $(L,m)$ in time at most $c_1t(|x|)$ using space at most $c_2\ell(m(x))$ on all inputs $x\in\Sigma^*$, where $c_1$ and $c_2$ are appropriate positive constants.  In our setting, $N$ is equipped with a read-only input tape and a semi-infinite rewritable work tape. Let $B$ denote a unique blank symbol of $N$.

We first modify $N$ so that it halts after making all work-tape cells blank and that it halts in scanning both $\cent$ on the input tape and $B$ in the \emph{start cell} (i.e., cell $0$) of the work tape.
Since $t$ is log-space time constructible, by modifying $N$ appropriately, we can make it halt in exactly $c'_1t(|x|)$ steps for an appropriate constant $c'_1$. Moreover, we can make $N$ have a unique accepting state. This last modification can be done by replacing all accepting states with a new unique accepting state, say, $q_{acc}$ and by making appropriate changes to all transitions.
For readability, we hereafter denote $c'_1$ by $c_1$.

In what follows, we wish to simulate $N$ by an $\dl$-uniform  family $\{M_{n,l}\}_{n,l\in\nat}$ of appropriate 2afa's specified by the proposition.   As a preparation, let $\tilde{\Gamma}=\Gamma\cup \{\hat{\sigma}\mid \sigma\in\Gamma\}$, where $\hat{\sigma}$ is a distinguished symbol indicating that the work-tape head is scanning symbol $\sigma$. The use of $\tilde{\Gamma}$ helps us simplify the description of the proof.

Let $x$ be any instance to $L$ and set $n=m(x)$. Let us consider \emph{surface configurations} $(q,j,k',w)$ of $N$ on $x$, each of which indicates that $N$ is in state $q$, scanning both the $j$th cell of the input tape and the $k'$th cell of the work tape composed of $w$.
We want to trace down these surface configurations of $N$ using an alternating series of universal states and existential states of $M_{n,|x|}$.
The number of all surface configurations is at most $c_1c_2|Q|\ell(n)  |\tilde{\Gamma}|^{\ell(n)}$ since each surface configuration belongs to $Q\times [0,|x|+1]_{\integer} \times [0,c_2\ell(n)-1]_{\integer} \times \tilde{\Gamma}^{\ell(n)}$.

Since each move of $N$ affects at most $3$ consecutive cells of its work tape, it suffices to focus our attention on these 3 local cells. Our idea is to define $M_{n,|x|}$'s surface configuration $((q,i,k,u),j)$ so as to  represent $N$'s surface configuration $(q,j,k',w)$ at time $i$ in such a way that $u$ indicates either the $k$th cell content or the content of its neighboring $3$ cells. Furthermore, when $k=k'$, $u$ also carries extra information (by changing tape symbol $\sigma$ to $\hat{\sigma}$) that the tape head is located at the $k$th cell.

Let us formally define the desired 2afa family $\{M_{n,l}\}_{n,l\in\nat}$ that computes $L$. We make $\tilde{\Gamma}_3$ composed of all $bcd$ satisfying that only \emph{at most one} of $b,c,d$ is in $\tilde{\Gamma}-\Gamma$ and the others are in $\Gamma$.
Consider any input $x$ and set $n=m(x)$.
An inner state of $M_{n,l}$ is of the form $(q,i,k,u)$ with $q\in Q$, $i\in[1,c_1t(|x|)]_{\integer}$, $k\in[0,c_2\ell(n)-1]_{\integer}$, and $u\in\tilde{\Gamma}\cup\tilde{\Gamma}_3$.
The number of such inner states is thus at most $2c_1c_2|Q|t(|x|)\ell(n)|\tilde{\Gamma}_3|\leq c_4t(|x|)\ell(n)$ for an appropriate constant $c_4>0$.
A \emph{surface configuration} of $M_{n,|x|}$ is a tuple $((q,i,k,u),j)$, where $(q,i,k,u)$ is an inner state of $M_{n,|x|}$ and $j\in[0,|x|+2]_{\integer}$.
This tuple $((q,i,k,u),j)$ expresses the following circumstance: at time $i$, $N$ is in state $q$, and $N$'s input-tape head is scanning cell $j$.
Moreover, if $u=a\in\Gamma$, then $N$'s work tape contains symbol $a$ in cell $k$ and its tape head is not scanning this cell.
In contrast, when $u=\widehat{a}$ with $a\in\Gamma$,
$N$'s tape head is scanning $a$ at cell $k$. Consider the case where $u=bcd\in\tilde{\Gamma}_3$.
If $bcd\in\Gamma^3$, then the cells indexed by $k-1,k,k+1$ respectively contain $b,c,d$ but $N$'s tape head does not stay on these cells. If $b=\widehat{\sigma_b}\in\tilde{\Gamma}-\Gamma$, then $N$ is scanning $\sigma_b$ at cell $k-1$. The other cases of $c=\widehat{\sigma_c}\in\tilde{\Gamma}-\Gamma$ and $d=\widehat{\sigma_d}\in\tilde{\Gamma}-\Gamma$ can be similarly treated.

Hereafter, we will describe how to simulate $N$'s computation on the machine $M_{n,|x|}$ by tracing down surface configurations of $N$ on $x$ using a series of universal and existential states of $M_{n,|x|}$.
Starting with $\gamma_0$, we inductively generate the next surface configuration of $M_{n,|x|}$ roughly in the following way. In an existential state, $M_{n,|x|}$ guesses (i.e., nondeterministically chooses) the content of $3$ consecutive cells in the current configuration of $N$ on $x$. In a universal state, $M_{n,|x|}$ checks whether the guessed content is indeed correct by branching out $3$ computation paths, each of which selects one of the $3$ chosen cells. The $c_2\ell(n)$-narrowness comes
from the space bound of $N$.

Let us return to a formal description of $M_{n,|x|}$. We first introduce a notation $\vdash$, which indicates a single transition of $N$ in terms of $M_{n,|x|}$'s surface configurations.
Let  $\gamma =((q,i,k,a),j)$  and  $\gamma' = ((p,i-1,k,bcd),j')$ be two surface configurations of $M_{n,|x|}$, where $i\geq1$, $p,q\in Q$, $a\in\tilde{\Gamma}$, $bcd\in\tilde{\Gamma}_3$, and $j,j'\in[0,|x|+1]_{\integer}$. We write $\gamma'\vdash \gamma$ if there exist constants $f,h\in\{\pm1\}$ and a symbol $e\in\Gamma$ that satisfy $j=j'+f$ and the following conditions (i)--(ii). (i) In the case of $a=\widehat{\sigma_a}\in\tilde{\Gamma}-\Gamma$, it holds that either both $b=\widehat{\sigma_b}\in\tilde{\Gamma}-\Gamma$ and $\delta(p,x_j,\sigma_b) = (q,\sigma_a,f,+1)$, or  both $d=\widehat{\sigma_d}\in\tilde{\Gamma}-\Gamma$ and $\delta(p,x_j,\sigma_d) = (q,\sigma_a,f,-1)$. (ii) In the case of $a\in\Gamma$, it holds that, if $b=\widehat{\sigma_b}\in\tilde{\Gamma}-\Gamma$, then $a=c$ and $\delta(p,x_j,\sigma_b)=(q,e,f,-1)$; if $d=\widehat{\sigma_d}\in\tilde{\Gamma}-\Gamma$, then $a=c$ and $\delta(p,x_j,\sigma_d)=(q,e,f,+1)$; and  if $c=\widehat{\sigma_c}\in\tilde{\Gamma}-\Gamma$, then $\delta(p,x_j,\sigma_c)=(q,a,f,h)$.
Using $\vdash$, we define $NEXT_{\gamma}^{(x)}$ to be the set of all surface configurations $\gamma'$ of the form
$((p,i-1,k,bcd),j')$ with $p\in Q$ and $bcd\in\tilde{\Gamma}_3$ satisfying  $\gamma'\vdash \gamma$. Note that $|NEXT_{\gamma}^{(x)}|\leq 2|Q||\tilde{\Gamma}|^3$ since only three parameters $(p,bcd,j')$
in $\gamma'$ may vary. In what follows, we set $i$ to be any number in $[0,c_1t(|x|)-1]_{\integer}$.

\renewcommand{\labelitemi}{$\circ$}
\begin{itemize}\vs{-2}
  \setlength{\topsep}{-2mm}%
  \setlength{\itemsep}{1mm}%
  \setlength{\parskip}{0cm}%

\item[(a)] The \emph{initial surface configuration} $\gamma_0$ of $M_{n,|x|}$ on $x$ is $((q_{acc},c_1t(|x|),0,\cent),0)$, which partly corresponds to the final accepting surface configuration $(q_{acc},0,0,\cent B\cdots B)$ of $N$ on $x$. We assign an $\exists$-label to this surface configuration. This is the $0$th step of the computation of $M_{n,|x|}$ on $x$.

\item[(b)] After step $2i$, we assume that the current surface configuration of $M_{n,|x|}$ is $\gamma =((q,i',k,a),j)$ with $i\geq1$, $q\in Q$, and $a\in\tilde{\Gamma}$, where $i'=c_1t(|x|)-i$. Note that the inner state $(q,i',k,a)$ has a $\forall$-label. At step $2i+1$, we nondeterministically choose one element $\gamma'= ((p,i'-1,k,bcd),j')$ from $NEXT_{\gamma}^{(x)}$. The inner state $(p,i'-1,k,bcd)$ has a  $\forall$-label. We assign ACCEPT to $\gamma$ if there is a surface configuration in $NEXT_{\gamma}^{(x)}$ whose label is ACCEPT.

\item[(c)] After step $2i+1$, letting $i'=c_1t(|x|)-i$, we assume that the current surface configuration is $\gamma=((p,i',k,bcd),j)$ with $bcd\in\tilde{\Gamma}_3$. The corresponding inner state $(p,i',k,bcd)$ is assumed to have a $\forall$-label. At step $2i+2$, we universally generate the following three surface configurations: $((p,i',k-1,b),j),((p,i',k,c),j),((p,i',k+1,d),j)$, without moving $M_{n,|x|}$'s tape head. Three new inner states $(p,i',k-1,b)$, $(p,i',k,c)$, and $(p,i',k+1,d)$ all have $\exists$-labels. We assign ACCEPT to $\gamma$ if the above three surface configurations are all labeled with ACCEPT.

\item[(d)] Leaves are of the form $((q,1,k,u),j)$ obtained by (c) after $2(c_1t(|x|)-1)$ steps with $q\in Q$, $0\leq k < c_2\ell(n)$, and $u\in \tilde{\Gamma}$. We assign ACCEPT to $((q,1,k,u),j)$ if $q=q_0 \wedge j=0$ and either $k=j \wedge u=\hat{B}$ or $k\neq j\wedge u=B$. Otherwise, we assign REJECT. This requirement comes from the fact that the $N$'s unique initial surface configuration is of the form $((q_{0},0,0,\cent B\cdots B)$.
\end{itemize}

The $\dl$-uniformity of $\{M_{n,l}\}_{n,l\in\nat}$ follows from its construction. Next, let us claim the following statement (*) for each fixed input $x$ with $m(x)=n$. Let $w_k$ denote the $(k+1)$th symbol of $\cent w$; in particular, $w_0=\cent$.
\begin{quote}
(*) For any tuple $(q,i,j,k',w)$ with $i\in[0,c_1t(|x|)-1]_{\integer}$, $j\in[0,c_2\ell(n)-1]_{\integer}$, and $w\in\cent\Gamma^{c_2\ell(n)}$,  letting $i'=c_1t(|x|)-i$, any surface configuration $(q,j,k',w)$ of $N$ on $x$ is in an accepting computation path after step $i$ iff all surface configurations $((q,i',k,w_{k}),j)$ of $M_{n,|x|}$ on $x$ for any $k\in[0,c_2\ell(m(x))]_{\integer}-\{k'\}$ are labeled by ACCEPT, and $((q,i',k',w_{k'}),j)$ is labeled with ACCEPT.
\end{quote}
This statement implies that there exists an accepting computation subtree of $M_{n,|x|}$ on $x$ iff $N$ has an accepting computation path on $x$. By the simulation of $N$, the height of the shortest accepting computation subtree is bounded from above by $2c_1t(|x|)-1$. We set $c_3=2c_1$.

Let us consider a \emph{computation graph} of $M_{n,|x|}$. The number of distinct surface configurations appearing at time $i$ of the computation graph of $M_{n,|x|}$ on $x$ is at most $c_2|Q|\ell(n)|\tilde{\Gamma}|^3$, which can be succinctly expressed as $c'_2\ell(n)$ using $c'_2=c_2|Q||\tilde{\Gamma}|^3$. Thus, $M_{n,|x|}$ is $c'_2\ell(n)$-narrow.

The remaining task is to verify Statement (*) by
{downward induction} on $i$.
Let us first consider the initial case of $i=c_1t(|x|)-1$. Notice that, when $N$ starts, its work tape is blank. By (d), each leaf-level surface configuration of $M_{n,|x|}$ is labeled by ACCEPT iff its corresponding tape symbol appears in the initial surface configuration of $N$ on $x$. Hence, Statement (*) is true for $i=c_1t(|x|)-1$. Assume by induction hypothesis that Statement (*) is true for an index $i\in[1,c_1t(|x|)-1]_{\integer}$.
For this index $i$, we set $i' =c_1t(|x|)-i$. Let us consider the case of $i-1$.
At a $\forall$-level, by (c), if  $((p,i',k-1,b),j), ((p,i',k,c),j), ((p,i',k+1,d),j)$ are all labeled by ACCEPT, then $((p,i',k,bcd),j)$ is also labeled by ACCEPT.
At an $\exists$-level, by (b), if $((p,i',k,bcd),j)$ is labeled by ACCEPT, then $((q,i'-1,k,a),j)$ is labeled by ACCEPT and if $N$ makes the correct move toward an accepting state with $((p,i',k,bcd),j) \vdash ((q,i'-1,k,a),j)$.
Combining both universal and existential levels, we conclude that $3$ consecutive cells obtained from cells indexed by $k-1$, $k$,and $k+1$ correctly represent a corresponding surface configuration of $N$. By mathematical induction, Statement (*) must be true.


To complete our proof, we still need to transform $M_{n,l}$ into another equivalent 2afa, say, $M'_{n,l}$ \emph{with no stationary move}.

\begin{claim}\label{no-stationay-move}
There is a 2afa $M'_{n,l}$ with no stationary move that simulates $M_{n,l}$ with $|Q||\check{\Sigma}|$ extra inner states.
\end{claim}

\begin{proof}
Let $\delta'$ be a transition function of $M_{n,l}$. The desired 2afa $M'_{n,l}$ simulates a stationary move of $M_{n,l}$ by stepping leftward and then moving back. More formally, we define a transition function $\hat{\delta}$ of $M'_{n,l}$ in the following way. First, we add new inner states of the form $\bar{q}_{\sigma}$ for each $q\in Q$ and $\sigma\in\check{\Sigma}$. If $(q,+1)\in\delta'(p,\sigma)$ holds, then $\hat{\delta}(p,\sigma)$ contains $(q,+1)$.
If $(q,0)\in\delta'(p,\sigma)$ holds, then we set $(\bar{q}_{\sigma},-1)\in\hat{\delta}(p,\sigma)$, $(q,+1)\in\hat{\delta}(\bar{q}_{\sigma},\tau)$ for any symbol $\tau\in\check{\Sigma}$.

Finally, we note that, from the proof of Claim \ref{no-stationay-move}, since  $\{M_{n,l}\}_{n,l\in\nat}$ is $\dl$-uniform,  $\{M'_{n,l}\}_{n,l\in\nat}$ is also $\dl$-uniform.
\end{proof}


(2) We assume that there are constants $c_1,c_2,c_3>0$ and an $\dl$-uniform family $\MM=\{M_{n,l}\}_{n,l\in\nat}$ of $c_2\ell(n)$-narrow 2afa's of at most $c_1 t(l)\ell(n)$ states such that $M_{m(x),|x|}$ computes $L(x)$ within $c_3t(|x|)$ time on all inputs $x$ given to $L$. Since $\MM$ is $\dl$-uniform, there is a log-space DTM $D$ that produces $\pair{M_{n,l}}$ from $1^n\#1^l$ for any pair $n,l\in\nat$.

We want to simulate $\{M_{m(x),|x|}\}_{n\in\nat}$ on the desired DTM, say, $N$.
This machine $N$ starts with an input $x$, computes $n=m(x)$, and runs $D$ on $1^n\#1^{|x|}$ to reconstruct $\pair{M_{n,|x|}}$ in $|x|^{O(1)}$ time using $O(\log(n+|x|))$ space.
Next, let us describe how to accept $x$ by $M_{n,|x|}$. Assume that $M_{n,|x|}$ has the form $(Q,\Sigma,\{\cent,\dollar\},\delta,q_0,Q_{\forall}, Q_{\exists},Q_{acc},Q_{rej})$. A surface configuration of $M_{n,|x|}$ on an  input $x$ has the form $(q,j)$ for $q\in Q$ and $j\in[0,|x|+1]_{\integer}$. The total number of such surface configurations is $|Q|(|x|+2)$, and thus it requires only $O(\log|x|)$ bits to describe each surface configuration.
For the simulation of $M_{n,|x|}$ on $x$, we use 3 work tapes. Initially, we write the initial surface configuration $(q_0,0)$ on the 1st work tape.  The 3rd work tape is used to remember ``time'' starting with $0$, which is initially written on this tape.

Let us consider a \emph{computation graph} of $M_{n,|x|}$ on $x$. Since $M_{n,|x|}$ runs in $c_3t(|x|)$ time, if $x$ is accepted, then there is an accepting computation subgraph of $M_{n,|x|}$ on $x$ having height at most $c_3t(|x|)$.
To check whether there is such an accepting computation subgraph, we perform a \emph{breadth-first traversal} of the computation graph.
In the first phase, starting with $i=0$, we recursively increment $i$ by one. This number $i$ indicates the current time in $[0,c_3t(|x|)]_{\integer}$. At  time $i\geq1$, we  replace each surface configuration $(q,j)$ by all surface  configurations reachable from it.
We denote by $C_i$ the set of all surface configurations $(q,j)$ that is reached by $M_{n,|x|}$ on $x$. By the $c_2\ell(n)$-narrowness of $M_{n,|x|}$, it follows that $|C_i|\leq c_2\ell(n)$. This replacing process continues until we reach $i=c_3t(|x|)$. To each surface configuration $(q,j)$ in $C_{t_3t(|x|)}$, we assign ACCEPT if $q$ is in $Q_{acc}$, and we assign REJECT otherwise.

In the second phase, we eliminate all surface configurations labeled by REJECT. We write $c_3t(|x|)$ on the 3rd tape. We start with $i=c_3t(|x|)$ and recursively decrease $i$ down to $1$. At stage $i$, we define $LIST_{i}$ to be the list of all surface configurations labeled by ACCEPT written on the 1st work tape. As in the first phase, assume that we have already generated
$C_{i-1}$ on the 2nd work tape. Let $(p,k)$ be any element of $C_{i-1}$.
If $M_{n,|x|}$ makes an existential move from $(p,k)$ to a set, say,  $\{(q_1,j_1),(q_2,j_2),\ldots,(q_c,j_c)\}$, then we assign ACCEPT to $(p,k)$ if at least one element of this set appears on the 1st work tape (thus it is labeled by ACCEPT). When $M_{n,|x|}$ makes a universal move from $(p,k)$ to a set $\{(q_1,j_1),(q_2,j_2),\ldots,(q_c,j_c)\}$,  we assign ACCEPT to $(p,k)$ if all elements of the set appear in $LIST_i$ (thus they are all labeled by ACCEPT). We then remove from $C_{i-1}$ all elements not labeled by ACCEPT. We overwrite all the remaining elements onto the 1st work tape and the obtained  list becomes $LIST_{i-1}$.  Finally, when reaching $i=1$, if $(q_0,0)$ is labeled by ACCEPT, then we accept $x$; otherwise, we reject $x$.

The 1st and the 2nd work tapes both use space at most $c_2\ell(n)\cdot O(\log|Q|(|x|+2))$, which coincides with $O(\ell(m(x))\log{|x|})$ because  $|C_i|\leq c_2\ell(m(x))$ and $O(\log{|x|})$ bits are needed to describe each surface configuration. The 3rd work tape uses $O(\log{t(|x|)})$ space. Thus, we can carry out the above procedure using space $O(\ell(m(x))\log|x|)+O(\log{t(|x|)})$.
The running time is at most $|x|^{O(1)}+ O(t(|x|)^2\ell(m(x))^2|x|^2)$.
Therefore, $(L,m)$ belongs to $\timespace(|x|^{O(1)}+(t(|x|)\ell(m(x))|x|)^2, \ell(m(x))\log{|x|}+\log{t(|x|)})$.
\end{proofof}

Proposition \ref{ptimespace-character} deals only with $M_{m(x),|x|}$ that correctly handles inputs taken from $\Sigma_n=\{x\in\Sigma^*\mid m(x)=n\}$. For a later application in Section \ref{sec:main-proof}, we want to extend the scope of such inputs to $\overline{\Sigma}_n=\{x\in\Sigma^*\mid m(x)\neq n\}$ by incorporating a family of simple 2dfa's that help us determine whether a given input is inside or outside of $\Sigma_n$.

We say that a size function $m$ is \emph{polynomially honest} (or \emph{p-honest}, for short) if there exists a constant $e>0$ satisfying $|x|\leq m(x)^e+e$ for all $x$.

\begin{lemma}\label{size-decide}
Assume that $m:\Sigma^*\to\nat^{+}$ is a p-honest log-space size parameter, where $\Sigma$ is an alphabet. Let $\Sigma_n=\{x\in\Sigma^*\mid m(x)=n\}$ and $\overline{\Sigma}_n=\{x\in\Sigma^*\mid m(x)\neq n\}$ for each index $n\in\nat$. There exist a constant $c>0$ and an $\dl$-uniform family $\{M_n\}_{n\in\nat}$ of simple 2dfa's such that, for every $n\in\nat$,  (1) $M_n$ solves $(\Sigma_n,\overline{\Sigma}_n)$ for all inputs in $\Sigma^*$ and (2) $M_n$ uses $n^{O(1)}$ inner states.
\end{lemma}

To prove the lemma, nonetheless, we require a supporting lemma stated below.

\begin{lemma}\label{log-space-computable}
For any log-space computable function $f:\Sigma^*\to\Gamma^*$ over two alphabets $\Sigma$ and $\Gamma$, there is an $\dl$-uniform family $\{M_n\}_{n\in\nat}$ of simple 2dfa's equipped with output tapes such that each $M_n$ has $n^{O(1)}$ states and $M_{|x|}$ computes $f(x)$ in $|x|^{O(1)}$ time for all inputs  $x\in\Sigma^*$.
\end{lemma}

\begin{proof}
Since $f$ is a log-space computable function from $\Sigma^*$ to $\Gamma^*$, we choose an appropriate DTM $N = (Q,\Sigma,\{\cent,\dollar\}, \Gamma,\delta,q_0,Q_{acc},Q_{rej})$ that computes $f$ in $|x|^{O(1)}$ time using $O(\log{|x|})$ space. Let $\Delta$ denote a work alphabet used for $N$'s work tape. Firstly, we translate surface configurations of $N$ to surface configurations of simple 2dfa's $\{M_n\}_{n\in\nat}$.
A surface configuration of $N$ on input $x$ is of the form $(q,j,k,w,a)$, which expresses that $N$ is in state $q$, scanning the $j$th input tape cell, the $k$th work tape cell containing $w$, and having $a$ written in the previous cell of $N$'s output tape.  We then define an inner state of $M_n$ as $(q,k,w,a)$ and denote by $Q_n$ the set of all such inner states. A surface configuration of $M_n$ on $x$ thus has the form $((q,k,w,a),j)$ for an index $j\in[0,|x|+1]_{\integer}$. From the work alphabet $\Delta$, we define $\tilde{\Delta} = \Delta\cup\{\hat{\sigma}\mid \sigma\in\Delta\}$.
We then introduce a binary relation $\vdash$ for two surface configurations $\gamma=((q,k,w,a),j)$ and $\gamma'=((p,k',w',a'),j')$ as follows: $\gamma\vdash \gamma'$ iff there are two directions $d_1,d_2\in\{\pm1\}$ for which
$\delta(q,x_j,w_k) = (p,\tau,a',d_1,d_2)$, $w'_k = \tau$, $rest_k(w)=rest_k(w')$, $j'=j+d_1\;(\mathrm{mod}\;|x|+2)$, and $k'=k+d_2\;(\mathrm{mod}\;|x|+2)$, where $rest_k(w)$ denotes the string obtained from $w$ by deleting the $k$th symbol $w_k$ of $w$.
Finally, we define a transition function $\delta_n$ of $M_n$ by setting  $\delta_{n}((q,k,w,a),x_j) = ((p,k',w',a'),d_1)$ exactly when $((q,k,w,a),j)\vdash ((p,k',w',a'),j+d_1\;(\mathrm{mod}\;|x|+2))$. For any input $x$, since $M_{|x|}$ simulates $N$ on $x$, $M_{|x|}$ computes $f(x)$.
The 2dfa $M_{|x|}$ runs in $|x|^{O(1)}$ time for any input $x$ and the number $|Q_n|$ of $M_{n}$'s inner states is at most $|Q|(c\log{n}+2)| |\Gamma| |\tilde{\Delta}|^{c\log{n}+2}$, which is upper-bounded by $n^{O(1)}$, for a certain constant $c\geq 1$.

To obtain the desired \emph{simple 2dfa's} $M'_n$, we next modify $M_n$ by adding extra $n^{O(1)}$ inner states to remember the tape head location. Here, we describe the behavior of $M'_n$ by sweeping rounds. At each sweeping round, starting at $\cent$, $M'_n$ sweeps the tape to locate the tape head of $M_n$ and, when $M'_n$ reaches the tape head location, it simulates $M_n$'s single move. More formally, $M'_n$'s transition function $\delta'_n$ is defined as follows. For simplicity, let $Q_n=\{q_0,q_1,\ldots,q_{t_n}\}$ with $t_n=|Q_n|$. We make an inner state of $M'_n$ have the form $(q,s,j)$ for $q\in Q_n$ and $s,j\in[0,n+2]_{\integer}$. Initially, we set $(q_0,0,0)$ to be a new initial state, provided that $q_0$ is the initial state of $M_n$. We then define a transition function $\delta'$ of $M'_n$ as $\delta'((q,s,j),\sigma) = (q,s,j+1,+1)$ if $s\neq j$ and $\sigma\neq\dollar$, $\delta((q,s,j),\dollar) = (q,s,0,+1)$ if $s\neq j$, and $\delta'((q,s,s),\tau) = (p,s+d\;(\mathrm{mod}\;n+2),s+1,+1)$ for $\tau\in\check{\Sigma}$. It then follows that $|Q'_n|$ is no more than $|Q_n|\cdot (c\log{n}+2)^2$, which is clearly at most $n^{O(1)}$.
Since $M'_n$ is sweeping, it is therefore simple, as requested.
\end{proof}

Next, we give the proof of Lemma \ref{size-decide}.

\begin{proofof}{Lemma \ref{size-decide}}
Let $m$ be any p-honest log-space size parameter and define $f(x) = 1^{m(x)}$ for all $x\in\{0,1\}^*$. Since this new function $f$ is log-space computable, Lemma \ref{log-space-computable} provides an $\dl$-uniform family $\{D_n\}_{n\in\nat}$ of $n^{O(1)}$-state simple 2dfa's such that, for any input $x$, $D_{|x|}$ computes $f(x)$ in $|x|^{O(1)}$ time. Since $m$ is p-honest, we take a constant $e>0$ satisfying $|x|\leq m(x)^e+e$ for all $x$.

We want to construct a new family $\{M_n\}_{n\in\nat}$ of simple 2dfa's solving $\{(\Sigma_n,\overline{\Sigma}_n)\}_{n\in\nat}$. Fix an index $n\in\nat$ and  define $M_n$ as follows. Letting $x$ be any input string over $\{0,1\}$, we run $D_{|x|}$ and try to produce $1^{m(x)}$ on its output tape. During this process, as soon as we discover that $m(x)>n$, we instantly reject $x$. This procedure takes $(n|x|)^{O(1)}$ steps. Next, assume that $m(x)\leq n$ and that we have already completed the production of $f(x)$ by running $D_{|x|}$ in $|x|^{O(1)}$ time. If $m(x)=n$, then we accept $x$; otherwise, we reject $x$. Since $D_{|x|}$ has $|x|^{O(1)}$ states, $M_n$ has state complexity of $(n|x|)^{O(1)}$, which is bounded by $n^{O(1)}$ because of $|x|\leq m(x)^{e}+e$ and $m(x)\leq n$.
\end{proofof}


We will give a nonuniform version of Proposition \ref{ptimespace-character} by making use of ``advice'' in place of the uniformity condition of machines in the proposition. In the proof of Proposition \ref{ptimespace-character}(1), for example, we use a DTM to produce $\pair{M_{n,l}}$ from $1^n\#1^l$; by contrast, in the nonuniform case, since we do not have such a DTM, we must generate $\pair{M_{n,l}}$ from information provided by a piece of given advice.

Let us recall that $\timespace(t(x),\ell(x,m(x)))/O(s(|x|))$ is an advised version of $\timespace(t(x),\ell(x,m(x)))$ with advice size $O(s(|x|))$.
Note that each underlying Turing machine characterizing a language in    $\timespace(t(x),\ell(x,m(x)))/O(s(|x|))$ is equipped with an additional read-only \emph{advice tape}, to which, for all input instances of each length $n$, we provide the advice tape with exactly one \emph{advice string} of length $O(s(n))$, surrounded by the two endmarkers $\cent$ and $\dollar$.

\begin{proposition}\label{nonuniform-ptimespace}
Let $t:\nat\to\nat^{+}$ be log-space time constructible and let $s,\ell:\nat\to\nat^{+}$ be $t(n)$-time space constructible.
Let $L$ and $m$ be a language and a log-space size parameter, respectively.
Define $h(l) = \max_{x:|x|=l}\{m(x)\}$ for each index $l\in\nat$.
\renewcommand{\labelitemi}{$\circ$}
\begin{enumerate}\vs{-2}
  \setlength{\topsep}{-2mm}%
  \setlength{\itemsep}{1mm}%
  \setlength{\parskip}{0cm}%

\item If $(L,m)\in \timespace(t(|x|),\ell(m(x))/O(s(|x|))$, then there is a nonuniform family $\{M_{n,l}\}_{n,l\in\nat}$ of $O(\ell(n))$-narrow 2afa's having $O(t(l)\ell(n)s(l))$ states such that $M_{m(x),|x|}$ computes $L(x)$ in $O(t(|x|))$ time on all inputs $x$ given to $L$.

\item If  there is a nonuniform family $\{M_{n,l}\}_{n,l\in\nat}$ of $O(\ell(n))$-narrow 2afa's having $O(t(l)\ell(n))$ states such that $M_{m(x),|x|}$ computes $L(x)$ in $O(t(|x|))$ time on all inputs $x$ to $L$, then $(L,m)$ belongs to $\timespace(|x|^{O(1)}+(t(|x|)\ell(m(x))|x|)^2, \ell(m(x))\log{|x|}+\log{t(|x|)}) /O(h(|x|)t(|x|)^2\log{t(|x|)})$.
\end{enumerate}
\end{proposition}

\begin{proof}
Given a parameterized decision problem $(L,m)$ over an alphabet $\Sigma$, let $L_n=\{x\in L\mid m(x)=n\}$ and $\overline{L}_n=\{x\in\overline{L}\mid m(x)=n\}$ for each index $n\in\nat$.

(1) Consider a DTM $N$ and advice $\{a_n\}_{n\in\nat}$ of size at most $s(|x|)$ such that $N(x,a_{|x|})$ computes $L(x)$ in $O(t(|x|))$ time using $O(\ell(m(x)))$ space on any input $x\in\Sigma^*$. Remember that an advice string is given to a \emph{read-only} advice tape. In the same way as in the proof of Proposition \ref{ptimespace-character}(1), we can define narrow 2afa's $M_{m(x),|x|}$ that simulates  $N(x,a_{|x|})$, except that, whenever $N$ accesses bits in $a_{|x|}$ given on the advice tape, $M_{m(x),|x|}$ retrieves the same information directly from the 2afa's inner states that encode each bit of $a_{|x|}$. To be more precise, we define inner states of $M_{m(x),|x|}$ to be $(q,i,k,l,u)$, where $(q,i,k,u)$ is defined similarly to the proof of Proposition \ref{ptimespace-character}(1) and $l$ refers to the location of a tape head on the advice tape. Thus, $M_{m(x),|x|}$ uses $c_1t(|x|)\ell(m(x))s(|x|)$ states since $l$ varies over $[0,s(|x|)+1]_{\integer}$.
A transition of $N$ on $(x,s_{|x|})$ is of the form $\delta(p,x_j,\sigma_a,s_l) = (q,c,f,e,g)$, where $g\in\{\pm1\}$ refers to the direction of an advice tape head. Since $\{a_{n}\}_{n\in\nat}$ is fixed, $M_{m(x),|x|}$ can be completely defined from $N$.

(2) Let $w(l) = h(l)t(l)^2\log{t(l)}$ for brevity. Assume that there is a nonuniform family $\{M_{n,l}\}_{n,l\in\nat}$ of $O(t(l)\ell(n))$-state $O(\ell(n))$-narrow 2afa's such that each $M_{m(x),|x|}$ computes $L(x)$ in $O(t(|x|))$ time on all inputs $x$.
Let $x$ be any input given to $L$ and set $n=m(x)$.
We define an advice string $s_{|x|}$ to be $\pair{\pair{M_{0,|x|}}\#\pair{M_{1,|x|}}\#\ldots \#\pair{M_{h(|x|),|x|}}}_2$. Note that there is a constant $c_4\geq1$ for which $|\pair{M_{i,|x|}}|\leq c_4 t(|x|)^2\log{t(|x|)}$ for all indices $i\in[0,h(|x|)]_{\integer}$.  It then follows that $|s_{|x|}|\leq c_4(h(|x|)+1) t(|x|)^2 \log{t(|x|)} \leq 2c_4 w(|x|)$. Thus, each advice string $s_n$ is of size $O(w(n))$.

Next, we intend to solve $(L,m)$ deterministically with the help of the above advice $\{s_n\}_{n\in\nat}$.
Starting with input $x$, compute $n=m(x)$ using log space and discover the encoding $\pair{M_{n,|x|}}$ inside of $s_{|x|}$ since $n\leq h(|x|)$.
Once we retrieve $\pair{M_{n,|x|}}$, we simulate $M_{n,|x|}$ on $x$ deterministically in the same way as in the proof of Proposition \ref{ptimespace-character}(2). This simulation of $M_{n,|x|}(x)$ requires  $O(t(|x|)^2\ell(n)^2|x|^2)$ time using space $O(\ell(n)\log{|x|})+O(\log{t(|x|)})$. Therefore, $(L,m)$ belongs to $\timespace(|x|^{O(1)}+(t(|x|)\ell(m(x))|x|)^2, \ell(m(x))\log|x|+\log{t(|x|)}) /O(w(x))$.
\end{proof}

\subsection{Automata Characterizations of 3DSTCON}\label{sec:automata-3DSTCON}

The proofs of Theorems \ref{3DSTCON-char-uniform}--\ref{3DSTCON-char-nonunif}  require a characterization of $3\dstcon$ with the size parameter $m_{ver}$ in terms of simple 2nfa's. Even though Kapoutsis and Pighizzini \cite{KP15} earlier gave a 2nfa-characterization of $\dstcon$, $3\dstcon$ needs a different characterization. This difference seems to be essential in the characterization of LSH because we still do not know whether $3\dstcon$ in the definition of LSH can be replaced by $\dstcon$; namely, $(3\dstcon,m_{ver})\notin\psublin$ iff $(\dstcon,m_{ver})\notin\psublin$. See \cite{Yam17a,Yam17b} for a discussion.

First, we re-formulate the parameterized decision problem $(3\dstcon,m_{ver})$ as a family $3\mathcal{DSTCON} = \{(3\dstcon_n,\overline{3\dstcon}_n)\}_{n\in\nat}$ of promise decision problems, each $(3\dstcon_n,\overline{3\dstcon}_n)$ of which is limited to directed graphs of vertex size exactly $n$, where $n$ refers to the value of the size parameter $m_{ver}$.
To express instances given to $(3\dstcon_n,\overline{3\dstcon}_n)$, we need to define an appropriate binary encoding of degree-bounded directed graphs.
Formally, let $K_n=(V,E)$ denote a unique \emph{complete directed graph} with $V=\{0,1,\ldots,n-1\}$ and $E=V\times V$ and let $G=(V,E)$ be any degree-$3$ subgraph of $K_n$.
We express this directed graph $G$ as its associated \emph{adjacency list}, which is represented by an $n\times 3$ matrix whose rows are indexed by $i\in[n]$ and columns are indexed by $j\in\{1,2,3\}$. Consider the $i$th row. Assume that there are vertices $j_{i,1},j_{i,2},\cdots j_{i,k}$ for which $G$ has edges $\{(i,j_{i,1}), (i,j_{i,2}), \cdots, (i,j_{i,k})\}$, provided that the vertex $i$ has outdegree $k\in[0,3]_{\integer}$ satisfying that   $j_{i,1}<j_{i,2}<\cdots <j_{i,k}$. To make the $i$th row contain $3$ entries, whenever $k<3$,  we automatically set $j_{i,t}$ to be the designated symbol $\bot$ for each index $t\in[k+1,3]_{\integer}$. The $i$th row of the adjacency list thus has a series $(j_{i,1},j_{i,2},j_{i,3})$ of exactly $3$ entries. For example, if $G$ has only two outgoing edges $(2,3)$ and $(2,5)$ from the vertex $2$, the corresponding list is $(3,5,\bot)$.

We further encode an adjacency list of $G$ into a single binary string, denoted by $\pair{G}$, in the following manner.
Let us recall the notation $binary(i)$ from  Section \ref{sec:numbers}.
Here, we extend this notation to $binary^*(i)$, which expresses $binary(i)$ if $i\in\nat$, and $\lambda$ (empty string) if $i=\bot$.
For each row indexed by $i$, if the row contains three entries $(j_{i,1},j_{i,2},j_{i,3})$ in the adjacency list, then this row is encoded as $C_i = binary(i)\# binary^*(j_{i,1})\# binary^*(j_{i,2})\# binary^*(j_{i,3})$, where $\#$ is a designated symbol not in $\{0,1\}$.
The \emph{binary encoding} $\pair{G}$ of $G$ is of the form $\pair{C_{1}\#^2C_{2}\#^2\cdots \#^2C_{n}}_2$.
Let $\Sigma_n$ be composed of all such encodings $\pair{G}$ of subgraphs $G$ of $K_n$.
Note that the bit size of this encoding $\pair{G}$ is $O(n\log{n})$ since $|binary(i)|,|binary^*(j_{i,l})|\leq c\log{n}+c$ for a certain constant $c>0$.

The next lemma asserts that we can easily check whether a given string is an binary encoding of a directed graph. This lemma helps us eliminate any \emph{invalid} instance easily and becomes a basis to the proof of Proposition \ref{construct-3DSTCON}.

\begin{lemma}\label{check-graph-code}
There exists an $\dl$-uniform family $\{N_n\}_{n\in\nat}$ of $O(n\log{n})$-state simple 2dfa's, each $N_n$ of which checks whether any given input $x$ is an encoding $\pair{G}$ of a certain degree-3 subgraph $G$ of $K_n$ in $O(n|x|)$ time; that is, within $O(n|x|)$ time, $N_n$ accepts $x$ if $x\in\Sigma_n$ and rejects $x$ if $x\notin\Sigma_n$.
\end{lemma}

\begin{proof}
The desired 2dfa $N_n$ works as follows. On input $x$, by sweeping an input tape repeatedly, $N_n$ checks the following four statements. (i) $x$ is of the form $z_1\pair{\#^2}_2z_2\pair{\#^2}_2z_3\cdots \pair{\#^2}_2z_k$ for certain even-length strings $z_1,z_2,\ldots,z_k\in\{0,1\}^*$, where $\pair{\#^2}_{2}=1^4$. (ii) $k=n$. (iii) For each $i\in[k]$, $z_i$ is of the form $\pair{binary(i)\# w_1\# w_2\# w_3}_2$ for even-length strings $w_1,w_2,w_3\in\{0,1\}^{\leq 4\ceilings{\log(n+1)}} \cup\{\bot\}$. (iv) Each $w_k$ is of the form $binary^*(j_{i,k})$ for a certain $j_{i,k}\in[n]$.
By the definition of $N_n$, it is a simple 2dfa.
To check (i)--(ii), we need at most $4\ceilings{\log(n+1)}$ states to count the number up to $n$. To perform the checking of (iii), we need to sweep the tape $n$ times and, at each round $i$, we check whether $z_i$ is of the correct form using at most $O(\log{n})$ states. Therefore, $N_n$ requires $O(n\log{n})$ states. Since $N_n$ sweeps the input tape $n$ times, we need the runtime of  $O(n|x|)$.
\end{proof}

Formally, the family $3\mathcal{DSTCON} = \{(3\dstcon_n,\overline{3\dstcon}_n)\}_{n\in\nat}$ is defined as follows.

\s
Degree-3 Directed $s$-$t$ Connectivity Problem for Size $n$ $(3\dstcon_n,\overline{3\dstcon}_n)$:
\renewcommand{\labelitemi}{$\circ$}
\begin{itemize}\vs{-2}
  \setlength{\topsep}{-2mm}%
  \setlength{\itemsep}{1mm}%
  \setlength{\parskip}{0cm}%

\item Instance: an encoding $\pair{G}$ of a subgraph $G$ of the complete directed graph $K_n$ with vertices of degree (i.e., indegree plus outdegree) at most $3$.

\item Output: YES if there is a path from vertex $0$ to vertex $n-1$;  NO otherwise.
\end{itemize}
Notice that each instance $x$ to $(3\dstcon_n,\overline{3\dstcon}_n)$ must satisfy $m_{ver}(x)=n$. The family $3\mathcal{DSTCON}$ naturally corresponds to $(3\dstcon,m_{ver})$, and thus  we freely identify $(3\dstcon,m_{ver})$ with the family $3\mathcal{DSTCON}$.

\begin{lemma}\label{m-ver-ideal}
There is a constant $c>0$ such that $m_{ver}$ satisfies $m(x)\leq |x|\leq c m(x)\log{m(x)}$ for all inputs $x$ given to $3\dstcon$ with $|x|\geq2$. Thus, $m_{ver}$ is an ideal size parameter.
\end{lemma}

\begin{proof}
We fix any ``reasonable'' encoding $\pair{\cdot}$ discussed in Section \ref{sec:main-result}. Let $x=\pair{G,s,t}$ be such a ``reasonable'' encoding of any instance to $3\dstcon$. Let $G=(V,E)$ and set  $m_{ver}(x)=n$.  Clearly, $n\leq|\pair{G,s,t}|$ holds since the encoding $\pair{G,s,t}$ must contain the whole information on $n$
vertices of $G$.
Since $|V|=n$, both $s$ and $t$ are encoded into strings of length $O(\log{n})$. Note that the adjacency list of $G$ has only $O(n)$ entries. Thus, we can take an absolute constant $c>0$ for which $|\pair{G,s,t}|\leq cn\log{n}$. Therefore, any ``reasonable'' encoding of $G$ requires $O(n\log{n})$ bits.
\end{proof}

Next, we want to build an $\dl$-uniform family $\{N_n\}_{n\in\nat}$ of constant-branching simple 2nfa's that solves  $3\mathcal{DSTCON}$ for all valid instances. Recall that, for each index $n\in\nat$, $\Sigma_n$ denotes the set of all \emph{valid} encodings of input subgraphs of $K_n$ given to $(3\dstcon_n,\overline{3\dstcon}_n)$.

Our definition of the ``running time'' of a 2nfa (as well as a 2afa) in Section \ref{sec:model-two-way} reflects only accepting computation. In what follows,  we are concerned with the running time of rejecting computation as well. Given a set $C$ of inputs and a time bound $t$, we conveniently say that a 2nfa (as well as a 2afa) $M$ \emph{rejects all inputs in $C$ in $t(n,|x|)$ time} if, for any string $x\in C$, when $M$ rejects $x$, all computation paths of $M$ on $x$ terminate within $t(n,|x|)$ steps.

\begin{proposition}\label{construct-3DSTCON}
There exists a log-space computable function $g$ for which $g$ produces in $n^{O(1)}$ time from each $1^n$ an encoding $\pair{M_n}$ of a $3$-branching simple 2nfa $M_n$ of $O(n\log{n})$ states that solves $(3\dstcon_n,\overline{3\dstcon}_n)$
in $O(n|x|)$ time for all inputs $x$ in  $\Sigma_n=3\dstcon_n\cup\overline{3\dstcon}_n$.
Moreover, $M_n$ can reject all inputs $x$ outside of
$\Sigma_n$ in $O(n|x|)$ time.
\end{proposition}

\begin{proof}
Firstly, we intend to describe how to generate, for each index $n\in\nat^{+}$, an encoding $\pair{M_n}$ of a  3-branching simple 2nfa $M_n = (Q,\Sigma,\{\cent,\dollar\},\delta,q_0,Q_{acc},Q_{rej})$ that solves  $(3\dstcon_n,\overline{3\dstcon}_n)$ for all inputs in $\Sigma_n$.
Notice that our 2nfa has a circular tape and moves its tape head only to the right. Secondly, we will extend $M_n$ to reject all inputs not in $\Sigma_n$.

Let us choose any valid input $x$ to $(3\dstcon_n,\overline{3\dstcon}_n)$, which is an encoding $\pair{G}$ of a certain subgraph $G$ of $K_n$, in which each vertex has degree at most $3$.
We design $M_n$ so that it works \emph{round by round} in the following way.
\begin{quote}
In the first round, we set $v_0$ to be the vertex $0$ and we move the tape head rightward from $\cent$ to $\dollar$. Assume by induction hypothesis that, at round $i$ ($\geq0$), we have already chosen vertex $v_i$ and have moved the tape head to $\dollar$. Nondeterministically, we select an index $j\in\{1,2,3\}$ while scanning $\dollar$ and then deterministically search for a row indexed $i$ in an adjacency list of $G$. By moving the tape head only from the left to the right along the circular tape, we read the content of the $(i,j)$-entry of the list. If it is $\bot$, then reject $x$ immediately. Next, assume otherwise. If $v_{i+1}$ is the $(i,j)$-entry, then we update the current vertex from $v_i$ to $v_{i+1}$. As soon as we reach the vertex $n-1$, we immediately accept $x$ and halt.
If $M_n$ tries to visit more than $n$ vertices, then we reject $x$ instantly.
\end{quote}

At each round, $M_n$ requires $O(\log{n})$ states to find the location of an   $(i,j)$-entry since each vertex is expressed by $O(\log{n})$ bits. To execute the last part of the above procedure, we need to use a ``counter,''  which is implemented using  $O(\log{n})$ space. It is thus clear that the above procedure of generating $\{M_n\}_{n\in\nat}$ requires space $O(\log{n})$.
The obtained $M_n$ solves $(3\dstcon_n,\overline{3\dstcon}_n)$ for all inputs in $\Sigma_n$. This is because, for each $x\in 3\dstcon_n$, $M_n$ enters an accepting state along a certain accepting computation path within $n$ rounds, and, for any $x\in\overline{3\dstcon}_n$, all computation paths are rejected within $n$ rounds. Since $M_n$ sweeps the input tape only once during each round, $M_n$ on $x$ halts in $O(n|x|)$ steps.

The $3$-branching property of $M_n$ comes from the fact that $M_n$ makes only at most $3$ nondeterministic choices while scanning $\dollar$. Since $M_n$ always sweeps the tape from the left to the right, it must be simple.
We further modify $M_n$ so that it rejects all inputs not in $\Sigma_n$. This modification can be done by combining $M_n$ with an $O(n\log{n})$-state simple 2dfa, obtained by Lemma \ref{check-graph-code}, that checks whether or not a given input $x$ belongs to $\Sigma_n$ in $O(n|x|)$ time.
Therefore, the total number of inner states of the final machine is  $O(n\log{n})$ and its runtime is $O(n|x|)$, as requested.
\end{proof}


Let us consider the converse of Proposition \ref{construct-3DSTCON}.

\begin{proposition}\label{2nfa-to-subgraphs}
Let $c\in\nat^{+}$ be any constant and define $e=\ceilings{\log(c+1)}$. Let $d(n) = 2^e(n+2)+2^{e+1}-1$ for any $n\in\nat$. There exists a function $g$ such that, for every $c$-branching simple 2nfa $M$ with $n$ states, $g$ takes an input of the form $\pair{M}\# x$ and outputs an encoding $\pair{G_x}$ of a degree-3 subgraph $G_x$ of $K_{d(n)}$  satisfying that $M$ accepts $x$ if and only if $\pair{G_x}\in 3\dstcon_{d(n)}$. Moreover, $g$ is computed by a certain $n^{O(1)}$-state simple 2dfa with a write-only output tape running in $n^{O(1)}\cdot O(|x|)$ time for all inputs $x$.
\end{proposition}

\begin{proof}
Let $c,n>0$ and let $M$ be any $c$-branching $n$-state simple 2nfa $M$ over  an alphabet $\Sigma$. Note that $|\pair{M}|=O(n\log{n})$. Let $e=\ceilings{\log(c+1)}$ and define  $d(n) = 2^e(n+2)+2^{e+1}-1$ for all indices $n\in\nat$.
We first modify $M$ so that it has a unique accepting state.
This is done as described in the proof of
Proposition \ref{ptimespace-character}.
We further modify $M$ so that it enters such an accepting state only at scanning $\dollar$. This modification can be done by turning the current accepting state into a non-halting state, moving its tape head to $\dollar$, and entering a new accepting state. We denote by $M'$ the obtained 2nfa.
Note that $M'$ is $(c+1)$-branching and has $n+3$ states.
Let $M' =(Q,\Sigma,\{\cent,\dollar\},\delta,q_0,Q_{acc},Q_{rej})$ with $|Q|=n+3$.
For convenience, we assume that $\Sigma=\{\sigma_1,\ldots,\sigma_l\}$  with $l=|\Sigma|$, $Q = \{q_0,q_1,\ldots,q_{n+2}\}$, and $Q_{acc}=\{q_{n+2}\}$.

Let $x\in\Sigma^*$ be any input to $M'$.
Hereafter, we want to describe how to build the desired subgraph $G_x$ of $K_{d(n)}$ by sweeping repeatedly an input tape that contains $\pair{M'}\#x$. We call by a \emph{round} a single traversal of the input tape from $\cent$ to $\dollar$. Clearly, $|\pair{M'}|$ remains $O(n\log{n})$.
A straightforward idea of transforming a computation tree of $M'$ on $x$ into $G_x$ does not work because the number of surface configurations of $M'$ on $x$ is $|Q|(|x|+2)$. Instead, our subgraph $G_x = (V_x,E_x)$ imitates a succinct version of a computation graph of $M'$ on $x$.

We use vertex labels of the form $\pair{k,j,r}$ for $k\in\{0,1\}$, $j\in[0,n+2]_{\integer}$, and $r\in[0,2^{e}-1]_{\integer}$. To each label $\pair{k,j,r}$, we assign the number $2^{e}j+2r+k$.
Since $\pair{k,j,r}\leq d(n)$ holds, it suffices to set $V_x=[0,d(n)-1]_{\integer}$. Intuitively, vertex $\pair{1,j,0}$ represents a situation where $M'$ is in state $q_j$ scanning $\dollar$, except for the case of $j=n+2$, whereas vertex $\pair{0,j,0}$ indicates that $M'$ is scanning $\cent$ in state $q_j$.
For each index $j\in[0,n+1]_{\integer}$, we assign $q_j$ to two vertices $\pair{0,j,0}$ and $\pair{1,j,0}$ as labels and assign $q_{n+2}$ to
vertex $\pair{0,n+2,0}$.

In round $j$ ($\geq 0$), assume that $M'$ is in state $q_j$ scanning $\cent$.
After reading the string $\cent x$, if $M'$ deterministically enters state $q_k$ and also moves its tape head to $\dollar$, then we include to $G_x$ an edge from $\pair{0,j,0}$ to $\pair{1,k,0}$.
To make every vertex of $G_x$ have degree at most $3$, we need to introduce a subgraph, called $E_j$, defined as follows. The subgraph $E_j$ consists of $2^e$ vertices labeled by $\pair{1,j,0},\pair{1,j,2},\ldots,\pair{1,j,2^e-1}$, and it has edges from $\pair{1,j,r}$ to both $\pair{1,j,2r}$ and $\pair{1,j,2r+1}$ for each index $r\in[0,2^e-1]_{\integer}$.
In scanning $\dollar$, if $\delta(q_k,\dollar)$ is a set $\{(q_{i_1},+1),(q_{i_2},+1),\ldots,(q_{i_{c+1}},+1)\}$ for certain $c+1$ indices $i_1,i_2,\ldots,i_{c+1} \in [0,n+1]_{\integer}$, then we include a subgraph $E_k$ to $G_x$ and also add $c+1$ edges $(\pair{1,k,2^{e-1}},\pair{0,i_1,0}), (\pair{1,k,2^{e-1}+1},\pair{0,i_2,0}), \ldots, (\pair{1,k,2^{e-1}+c},\pair{0,i_{c+1},0})$.
The degree of each vertex is at most $3$, and thus $G_x$ is a degree-3 subgraph of $K_{d(n)}$. Although $G_x$ depends on $x$, the vertex size of $G_x$ is independent of $x$ but $|Q|$.
The number of rounds is at most $|Q|$. In each round, we need to read each data $(q_{i_t},+1)$ from the list by sweeping the tape and produces its corresponding edge; thus, we require time $(c+1)|x|\cdot O(\log{nd(n)})$, which equals $O(|x|\log{n})$.
By the definition of $G_x$, it follows that $M'$ accepts $x$ by entering state $q_{n+1}$ if and only if $G_x$ has a path from vertex $\pair{0,0,0}$ to vertex $\pair{0,n+2,0}$.
We thus obtain $|\pair{G_x}|=O(d(n)\log{d(n)})$. Since $d(n)=O(n)$, it also follows that  $|\pair{G_x}|=O(n\log{n})$.

Finally, we define $g$ to be a function that takes an input of the form $\pair{M'}\#x$ and produces $\pair{G_x}$ on an output tape.
Since the above procedure is deterministic and can be done by sweeping the input tape $n$ times as well as generating extra vertices and edges at each round, $g$ must be computed by a certain $n^{O(1)}$-state simple 2dfa running in $n^{O(1)}\cdot O(|x|)$ time for any input $x$.
\end{proof}

\section{Proofs of Theorems \ref{3DSTCON-char-uniform} and \ref{3DSTCON-char-nonunif}}\label{sec:main-proof}

We are ready to give the desired proofs of Theorems \ref{3DSTCON-char-uniform}--\ref{3DSTCON-char-nonunif} in the following  two  subsections.
In Section \ref{sec:generalization}, we will prove Theorems \ref{3DSTCON-char-uniform} and \ref{3DSTCON-char-nonunif}(1)--(2) in a more general setting. In Section \ref{sec:change-state-complexity}, we will show the remaining part (3) of Theorem  \ref{3DSTCON-char-nonunif}.

\subsection{Generalizations to PTIME,SPACE($\cdot$)}\label{sec:generalization}

Let us recall Theorems \ref{3DSTCON-char-uniform} and \ref{3DSTCON-char-nonunif}(1)--(2), which are concerned only with the parameterized complexity class $\psublin$. In fact, it is possible to prove more general theorems. We intend to present such theorems (Theorems  \ref{uniform-3DSTCON-general} and \ref{nonunif-general-3DSTCON}) for a general   parameterized complexity class $\ptimespace(s(x,m(x)))$ defined in Section \ref{sec:sub-linear-space}, which is the union of all classes  $\timespace(p(|x|),s(x,m(x)))$ for any positive polynomial $p$.

A set $\FF$ of functions $\ell:\nat\to\nat^{+}$ is said to be \emph{logarithmically saturated} if, for every function $\ell\in\FF$ and every constant $c>0$, there are two functions $\ell',\ell''\in\FF$ such that $\ell(\ceilings{cn\log{n}})\leq \ell'(n)$ and $c\ell(n)\log{n} \leq \ell''(n)$ for all numbers $n\in\nat$.

\begin{theorem}\label{uniform-3DSTCON-general}
Let $\FF$ denote an arbitrary nonempty logarithmically-saturated set and assume that every function in $\FF$ is polynomially bounded.
The following three statements are logically equivalent.
\renewcommand{\labelitemi}{$\circ$}
\begin{enumerate}\vs{-2}
  \setlength{\topsep}{-2mm}%
  \setlength{\itemsep}{1mm}%
  \setlength{\parskip}{0cm}%

\item There exists a function $\ell\in\FF$ such that $(3\dstcon,m_{ver})$ is in $\bigcup_{m}\ptimespace(\ell(m(x)))$, where $m$ ranges over all log-space size parameters.

\item For any constant $c>0$, there is a function $\ell\in\FF$ such that, for any constant $e>0$, every $\dl$-uniform family of $c$-branching simple 2nfa's with at most  $en\log{n}+e$ states is converted into another $\dl$-uniform family of  $n^{O(1)}$-state $O(\ell(n))$-narrow 2afa's that agree with them on all inputs.

\item For each constant $c\in\nat^{+}$, there are a function $\ell\in\FF$ and a log-space computable function $f$ such that
    $f$ takes an input of the form $\pair{M}$ for a $c$-branching $n$-state simple 2nfa $M$ and $f$ produces an encoding $\pair{N}$ of another $n^{O(1)}$-state $O(\ell(n))$-narrow 2afa $N$ that agrees with $M$ on all inputs.
\end{enumerate}\vs{-2}
Furthermore, fixing $c$ to $3$ in the above statements does not change their logical equivalence.
\end{theorem}

An advised version of Theorem \ref{uniform-3DSTCON-general} is given below.

\begin{theorem}\label{nonunif-general-3DSTCON}
Let $\FF$ be any nonempty polynomially-bounded logarithmically-saturated set of functions $\ell:\nat\to\nat^{+}$.
The statements 1 and 2 below are logically equivalent.
\renewcommand{\labelitemi}{$\circ$}
\begin{enumerate}\vs{-2}
  \setlength{\topsep}{-2mm}%
  \setlength{\itemsep}{1mm}%
  \setlength{\parskip}{0cm}%

\item  There is an $\ell\in\FF$ such that $(3\dstcon,m_{ver})$ is in  $\bigcup_{m}\ptimespace(\ell(m(x)))/\poly$, where $m$ ranges over all log-space size parameters.

\item For each constant $c\in\nat^{+}$,  there is a function $\ell\in\FF$ such that any $n$-state $c$-branching simple 2nfa can be converted into another $n^{O(1)}$-state $O(\ell(n))$-narrow 2afa that agrees with it on all inputs.
\end{enumerate}\vs{-2}
Furthermore, it is possible to fix $c$ to $3$ in the above statements.
\end{theorem}

From Theorems  \ref{uniform-3DSTCON-general} and \ref{nonunif-general-3DSTCON},  we can derive Theorems \ref{3DSTCON-char-uniform} and \ref{3DSTCON-char-nonunif}(1)--(2).

\vs{-3}
\begin{proofof}{Theorems \ref{3DSTCON-char-uniform} and \ref{3DSTCON-char-nonunif}(1)--(2)}
Consider all functions $\ell(n)$ of the form $\ell(n)=n^{\varepsilon}$ for  certain constants $\varepsilon\in[0,1)$. We define $\FF$ to be the collection of all such functions. It is not difficult to show that  $\FF$ is logarithmically saturated. By the definition of $\psublin$, it follows that   $(3\dstcon,m_{ver})$ is in $\bigcup_{m}\ptimespace(\ell(m(x)))$ for a certain $\ell\in\FF$ iff $(3\dstcon,m_{ver})$ is in $\psublin$. Theorem  \ref{uniform-3DSTCON-general} then leads to Theorem \ref{3DSTCON-char-uniform}. In a similar way, we can prove Theorem \ref{3DSTCON-char-nonunif}(1)--(2) using Theorem \ref{nonunif-general-3DSTCON}.
\end{proofof}

Now, we return to Theorem \ref{uniform-3DSTCON-general} and describe its proof, in which we partly utilize propositions and lemmas given in Section \ref{sec:basic-characterize}.

\vs{-2}
\begin{proofof}{Theorem \ref{uniform-3DSTCON-general}}
For convenience, given a function $\ell$, we write $\CC_{\ell}$ for the union $\bigcup_{m}\ptimespace(\ell(m(x)))$ taken over all log-space
size parameters $m$. For each index $n\in\nat$, let $\Sigma_n=\{x\mid m_{ver}(x)=n\}$ and $\overline{\Sigma}_n=\{x\mid m_{ver}(x)\neq n\}$.
Take any nonempty set $\FF$ that satisfies the premise of the theorem. Note that the last line of the theorem (namely, fixing $c$ to be $3$) can be derived by examining the following three parts of our proof.

[1 $\Rightarrow$ 3]
Assume that $(3\dstcon,m_{ver})\in\CC_{\ell}$ for a certain
function $\ell\in\FF$.
Since $(3\dstcon,m_{ver})$ is in $\timespace(|x|^s,\ell(m_{ver}(x)))$ for a certain constant $s\geq1$, Proposition \ref{ptimespace-character}(1) yields an $\dl$-uniform family $\{D_{n,l}\}_{n,l\in\nat}$ of $O(\ell(n))$-narrow 2afa's having $O(l^s\ell(n))$ states such that $D_{m_{ver}(x),|x|}$ computes  $3\dstcon(x)$ in $O(|x|^s)$ time on all inputs  $x$ given to $3\dstcon$.
We further extend $D_{m_{ver}(x),|x|}$ to work in $O(|z|^s)$ time for all inputs $z$ (even if $z$ is not a valid encoding of a graph) by combining an $\dl$-uniform family of $O(n\log{n})$-state $O(n|x|)$-time simple 2dfa's, given by Lemma \ref{check-graph-code}, which solves $\{(\Sigma_n,\overline{\Sigma}_n)\}_{n\in\nat}$ for all inputs.

By the $\dl$-uniformity, there is a log-space computable function $h$ that produces $\pair{D_{n,l}}$ from input $1^n\#1^l$.
We then convert this function $h$ into an $\dl$-uniform family of simple 2dfa's with at most $n^{b}$ states running in $|x|^b$ time for a certain constant  $b\in\nat^{+}$ by Lemma \ref{log-space-computable}.  Let $d(n)= 2^a(n+2) + 2^{a+1} - 1$, where $a=\ceilings{\log(c+1)}$.
Let $c$ be any constant in $\nat^{+}$. Proposition \ref{2nfa-to-subgraphs} provides a function $g$ such that, for any $c$-branching simple $n$-state 2nfa $M$, $g$ transforms $\pair{M}\# x$ to an encoding $\pair{G_x}$ of a degree-3 subgraph $G_x$ of $K_{d(n)}$ satisfying  that $M$ accepts $x$ exactly when $\pair{G_x}\in 3\dstcon_{d(n)}$. Note that $|\pair{G_x}| = O(d(n)\log{d(n)}) = O(n\log{n})$.
Moreover, $g$ is computed by a certain $L$-uniform family $\{E_n\}_{n\in\nat}$ of simple 2dfa's running in $n^{O(1)}\cdot O(|x|)$ time using at most $n^{t}$ states for a certain fixed constant $t\in\nat^{+}$.

To show Statement (3), it suffices to design a log-space computable function $f$  that transforms every  $c$-branching $n$-state simple 2nfa $M$ to another equivalent 2afa $N$ of the desired type, because the log-space computability of $f$ guarantees the $\dl$-uniformity of the obtained family of 2afa's. Here, we define $f$ so that, given an encoding $\pair{M}$ of each $c$-branching simple 2nfa $M$ with $n$ states, it produces an encoding $\pair{N}$ of an appropriate 2afa $N$ that works as follows.
\begin{quote}
On input $x$, prepare $\pair{M}$, generate $\pair{G_x}$ from $\pair{M}\# x$ of length $O(|\pair{M}|+|x|)$ by applying $g$, and compute  $l=|\pair{G_x}|$, which is $O(n\log{n})$. From $1^{d(n)}\#1^l$, produce $\pair{D_{d(n),l}}$ by applying $h$, and run $D_{d(n),l}$ on the input $\pair{G_x}$. Note that we cannot actually write down $\pair{G_x}$ onto $N$'s work tape; however, since $g$ is computed by the family $\{E_n\}_{n\in\nat}$, we can produce any desired bit of $\pair{G_x}$ by running such 2dfa's.
\end{quote}
It then follows by the definition of $N$ that $\pair{G_x}\in 3\dstcon_{d(n)}$ iff $N$ accepts $x$. From this equivalence, we conclude that $M$ accepts $x$ iff $N$ accepts $x$. Let us choose a function  $\ell'\in\FF$ satisfying $\ell(d(n))\leq \ell'(n)$ for all $n\in\nat$. Note that the total number of $N$'s inner states is at most $O(n^bl^s\ell(d(n)))$,
which equals $O(n^{s+b+1}\ell'(n))$ because $l=O(n\log{n})$. Since $\ell'$ is polynomially bounded by our assumption, $N$ have $n^{O(1)}$ states.

[3 $\Rightarrow$ 2]
Assuming Statement (3), for each constant $c\geq1$, we obtain a function $\ell\in\FF$ and a log-space computable function $f$ that, from any $c$-branching $n$-state simple 2nfa, produces an $O(\ell(n))$-narrow 2afa with $O(n^s)$ states that agrees with it on all inputs, where $s\geq1$ is a constant. To show Statement (2), let us take any $\dl$-uniform family $\{M_n\}_{n\in\nat}$ of $c$-branching simple 2nfa's, each $M_n$ of which has at most $en\log{n}$ states for an arbitrarily fixed constant $e\geq 0$. By the $\dl$-uniformity of $\{M_{n}\}_{n\in\nat}$, we choose a log-space DTM $D$ that produces $\pair{M_n}$ from $1^n$ for every index  $n\in\nat$.

Let us consider the following deterministic procedure, in which we generate a new 2afa, called $N_n$, from $1^n$. This procedure naturally introduces its corresponding function, which we call $g$.
\begin{quote}
On input $1^n$, we apply $D$ and produce $\pair{M_n}$. We then apply $f$ to $\pair{M_n}$ and obtain its equivalent 2afa $\pair{N_n}$.
\end{quote}
Since $\{M_n\}_{n\in\nat}$ is $\dl$-uniform, so is $\{N_n\}_{n\in\nat}$.
Moreover, since $M_n$ has at most $en\log{n}$ states, $N_n$  is $O(\ell(en\log{n}))$-narrow and has $O((en\log{n})^s)$ states. The state complexity of $N_n$ is therefore $O(n^{s+1})$.
Note that, by our assumption, we can choose a function $\ell'\in\FF$ satisfying  $\ell(en\log{n})\leq \ell'(n)$ for all $n\in\nat$. It then follows that $N_n$ is $\ell'(n)$-narrow. Finally, since $f$ and $D$ are log-space computable, so is $g$.

[2 $\Rightarrow$ 1] We start with Statement (2). We fix $c=3$. For a certain function $\ell\in\FF$, any $\dl$-uniform family of $3$-branching simple 2nfa's with at most $e'n\log{n}+e'$ states for a constant $e'>0$ can be converted into another $\dl$-uniform family of   $n^{O(1)}$-state $O(\ell(n))$-narrow 2afa that agrees with it on all inputs. Let us consider $3\mathcal{DSTCON}$. By Proposition \ref{construct-3DSTCON}, we obtain an  $\dl$-uniform family $\MM$ of $3$-branching simple 2nfa's with at most $en\log{n}$ states that solve $(3\dstcon_n,\overline{3\dstcon}_n)$ for all inputs, where $e>0$ is an appropriate constant.
By applying our assumption to $\MM$, we obtain an $\dl$-uniform family $\{N_n\}_{n\in\nat}$ of 2afa's, each $N_n$ of which has $n^{O(1)}$ states, is $O(\ell(n))$-narrow, and solves $3\mathcal{DSTCON}$ for all inputs.

Let us set $N_{n,l}$ to be $N_n$ for each pair $n,l\in\nat$. Note that, for any given input $x$, $N_{m_{ver}(x),|x|}$ computes $3\dstcon_{m_{ver}(x)}(x)$ in time $(m_{ver}(x))^{O(1)}\cdot O(|x|) \subseteq |x|^{O(1)}$, as noted in Section \ref{sec:model-two-way}, because $N_{n,l}$ has $n^{O(1)}$ states and $m_{ver}(x)\leq|x|$.
Proposition \ref{ptimespace-character}(2) then guarantees that $(3\dstcon,m_{ver})$ belongs to $\timespace(|x|^{O(1)},\ell(m_{ver}(x))\log{|x|})$. By Lemma \ref{m-ver-ideal},   for a constant $a>0$, we obtain  $|x|\leq am_{ver}(x)\log{m_{ver}(x)}$ for all inputs $x$ given to $3\dstcon$ with $|x|\geq2$. It thus follows that $\ell(n)\log{|x|}\leq \ell(n)\log(an\log{n})\leq 2\ell(n)\log{n}$.
Pick a function $\ell'\in\FF$ satisfying $2\ell(n)\log{n}\leq \ell'(n)$ for all numbers $n\in\nat$. Therefore, $(3\dstcon,m_{ver})$ is included in $\CC_{\ell'}$.
\end{proofof}

Theorem \ref{uniform-3DSTCON-general} also leads to Proposition \ref{Barnes-translate} on top of the result of Barnes \etalc~\cite{BBRS98} on $(\dstcon,m_{ver})$.

\begin{proofof}{Proposition \ref{Barnes-translate}}
Let us define $\FF$ to be the set composed of all functions of the form $\ell(n) = n^{1-c/\sqrt{\log{n}}}$ for each constant $c>0$. Note that $\FF$ is logarithmically saturated. The main result of Barnes \etalc~\cite{BBRS98} states that $(\dstcon,m_{ver})$ belongs to $\ptimespace(\ell(m_{ver}(x)))$ for a certain function $\ell\in\FF$.
This implies that
$(3\dstcon,m_{ver})$ is also in $\ptimespace(\ell(m_{ver}(x)))$ since $3\dstcon$ is a restricted variant of $\dstcon$. By Theorem \ref{uniform-3DSTCON-general}(2), there is a function $\ell\in\FF$ such that every $\dl$-uniform family of $c$-branching simple 2nfa's having $O(n\log{n})$ states can be converted into another $\dl$-uniform family of equivalent 2afa's, which are $O(\ell(n))$-narrow and have $n^{O(1)}$ states.
\end{proofof}

The proof of Theorem \ref{nonunif-general-3DSTCON} is in essence similar to that of Theorem \ref{uniform-3DSTCON-general} except for the treatment of advice strings.

\begin{proofof}{Theorem \ref{nonunif-general-3DSTCON}}
In what follows, for any given function $\ell\in\FF$, we succinctly write $\CC_{\ell}$ for $\ptimespace(\ell(m_{ver}(x)))/\poly$. By examining the following proof, we can show the last line of the theorem.

[1 $\Rightarrow$ 2] The argument below is similar to the one for [1 $\Rightarrow$ 3] in the proof of Theorem \ref{uniform-3DSTCON-general}. Assume that $(3\dstcon,m_{ver})$ is in $C_{\ell}$ for a certain function $\ell\in\FF$.
Take two polynomials $s$ and $t$ for which $\timespace(t(|x|),\ell(m_{ver}(x)))/O(s(|x|))$ contains $(3\dstcon,m_{ver})$.
By setting $L=3\dstcon$ and $m=m_{ver}$, Proposition \ref{nonuniform-ptimespace}(1) yields a family $\{D_{n,l}\}_{n,l\in\nat}$ of $O(t(l)\ell(n)s(l))$-state $O(\ell(n))$-narrow 2afa's such that $D_{m_{ver}(x),|x|}$ computes $3\dstcon(x)$ in $O(t(|x|))$ time on all inputs $x$ given to $3\dstcon$.
Combining this with Lemma \ref{check-graph-code}, we can assume without loss of generality that $D_{m_{ver}(x),|x|}$ rejects all inputs $x$ with $m_{ver}(x)\neq n$ in $O(m_{ver}(x)|x|)$ time.
Let $d(n) = 4(n+2)+7$ and take a function $\ell'\in\FF$ satisfying $\ell(d(n))\leq \ell'(n)$ for all $n\in\nat$.

Toward Statement (2), take any $n$-state $c$-branching simple 2nfa $M$. Similarly to the proof of Theorem \ref{uniform-3DSTCON-general}, let $g$ transform $\pair{M}\#x$ to an encoding $\pair{G_x}$ of an appropriate degree-3 subgraph $G_x$ of $K_{d(n)}$.
Now, we define $N$ that works as follows.
\begin{quote}
On input $x$, produce $\pair{M}\# x$, generate $\pair{G_x}$ from $\pair{M}\#x$ by applying $g$, and simulate $D_{d(n),|\pair{G_x}|}$ on $\pair{G_x}$.
\end{quote}
By the above definition, $N$ agrees with $M$ on all inputs $x$.  Note that $N$ has $O(t(|\pair{G_x}|)\ell'(n)s(|\pair{G_x}|))$ states and is $O(\ell'(n))$-narrow. Since $|\pair{G_x}|=O(n\log{n})$, we conclude that $N$ has state complexity of $n^{O(1)}$.

[2 $\Rightarrow$ 1] Assume Statement (2). This implies that (*) every $3$-branching $n$-state simple 2nfa can be converted into another $n^{O(1)}$-state $O(\ell(n))$-narrow 2afa that agrees with it on all inputs. Hereafter, we  describe how to solve  $(3\dstcon,m_{ver})$ with help of advice in polynomial time using $O(\ell(m_{ver}(x)))$ space.
As an immediate consequence of Proposition \ref{construct-3DSTCON}, we can take a constant $e>0$ and an $\dl$-uniform family $\{M_n\}_{n\in\nat}$ of $3$-branching simple 2nfa's with at most $en\log{n}$ states recognizing $3\mathcal{DSTCON}$ in $O(n|x|)$ time on all inputs $x$.

In order to use Statement (*), we need to modify $M_n$ into a new 2nfa $M'_{k(n)}$, where $k(n)=\ceilings{en\log{n}}$, by appending  $k(n)-p(n)$ extra dummy states, which neither contribute to the essential behavior of $M_n$ nor lead to any halting state. As a result, $M'_{k(n)}$ has exactly $k(n)$ states and it computes $3\dstcon_n(x)$ in $O(n|x|)$ time on all inputs $x$. By our assumption, there are constants $s,t\in\nat^{+}$ and an  $O(\ell(k(n)))$-narrow 2afa $N_{k(n)}$ with at most $k(n)^{s}+s$ states that agrees with $M'_{k(n)}$ in $O(k(n)|x|)$ time on all inputs $x$.
Take another function $\ell'\in\FF$ satisfying $\ell(k(n)) \leq \ell'(n)$ for all $n\in\nat$.

The desired 2afa $N'_n$ is defined to be $N_{k(n)}$. Note that $N'_n$ has $O(n\log{n})$ states, is $O(\ell'(n))$-narrow, and runs in $(n|x|)^{O(1)}$ time. By setting $N'_{n,l}$ as $N'_{n}$ for any $l\in\nat$, we can apply Proposition \ref{nonuniform-ptimespace}(2) to $\{N'_{n,l}\}_{n,l\in\nat}$, and we then conclude that $(3\dstcon,m_{ver})$ belongs to $\ptimespace(\ell'(m_{ver}(x))\log{|x|})/|x|^{O(1)}$, where $x$ indicates a symbolic input.  Note that $\ell'(n)\log|x|\leq 2\ell'(n)\log{n}$ because, for a certain absolute constant $e>0$, $|x|\leq em_{ver}(x)\log{m_{ver}(x)}$ holds for all valid inputs $x$. Choose another function $\ell''\in\FF$ such that $\ell'(n)\log{n}\leq \ell''(n)$ for all $n\in\nat$. With this $\ell''$, $(3\dstcon,m_{ver})$ is included in $\ptimespace(\ell''(m_{ver}(x)))/\poly$, which equals $\CC_{\ell''}$.
\end{proofof}

\subsection{Relationships among State Complexity Classes}\label{sec:change-state-complexity}

To complete the proof of Theorem \ref{3DSTCON-char-nonunif}, we  still need the verification of the logical equivalence between Statements (1) and (3) of the theorem. To achieve this goal, we first show Proposition \ref{PsubLIN-to-2A},  which asserts a close relationship between $\psublin$ and $\bigcup_{\varepsilon\in[0,1)} 2\mathrm{A}_{narrow(n^{\varepsilon})}$.

\begin{proposition}\label{PsubLIN-to-2A}
Given a parameterized decision problem $(L,m)$, let  $\LL=\{(L_{n},\overline{L}_{n})\}_{n\in\nat}$ be the family induced from $(L,m)$. Assume that $m$ is an ideal log-space size parameter.
It then follows that $(L,m)\in\psublin/\poly$ iff $\LL\in \bigcup_{\varepsilon\in[0,1)} 2\mathrm{A}_{narrow(n^{\varepsilon})}$.
\end{proposition}

\begin{proof}
For notational simplicity, we write $2\tilde{\mathrm{A}}$ to denote the union $\bigcup_{\varepsilon\in[0,1)} 2\mathrm{A}_{narrow(n^{\varepsilon})}$. Let us assume the premise of the proposition. Since $m$ is ideal, there are constants $c_1,c_2>0$ and $k\geq1$ for which $c_1 m(x)\leq |x|\leq c_2m(x)\log^k{m(x)}$ for all $x$ with $|x|\geq2$. Since $m$ is log-space computable, Lemma \ref{log-space-computable} provides a family $\{D_n\}_{n\in\nat}$ of 2dfa's of $n^{O(1)}$ states such that, for each input $x$, $D_{|x|}$ computes $1^{m(x)}$ in $|x|^{O(1)}$ time.

[If -- part] Assume that $\LL\in2\tilde{\mathrm{A}}$. Take a constant $\varepsilon\in[0,1)$ and a nonuniform family $\{M_{n}\}_{n\in\nat}$ of $O(n^{\varepsilon})$-narrow 2afa's with $n^{O(1)}$ states that solves $\LL$ for all inputs. For any pair $n,l\in\nat$, we define a machine $M_{n,l}$ to be $M_n$.
As noted in Section \ref{sec:model-two-way}, the height of any accepting computation tree of $M_{n,|x|}$ on input $x$ is upper-bounded by $n^{O(1)}\cdot (|x|+2)$. Therefore, $M_{m(x),|x|}$ computes $L(x)$ in time $m(x)^{O(1)}\cdot O(|x|)$, which is at most $|x|^{O(1)}$ since $c_1m(x)\leq |x|$.
By setting $t(n)=n^{O(1)}$ and $\ell(n)=n^{\varepsilon}$ in Proposition \ref{nonuniform-ptimespace}(2),  we conclude that $(L,m)$ belongs to $\timespace(|x|^{O(1)}m(x)^{\varepsilon}, m(x)^{\varepsilon}+O(\log{m(x)}))/O(h(m(x))|x|^{O(1)})$, which equals $\timespace(|x|^{O(1)},m(x)^{\varepsilon})/O(|x|^{O(1)})$, where $h(l)=\max_{x:|x|=l}\{m(x)\}$. The last complexity class is clearly included in $\psublin/\poly$.

[Only if -- part] Assume that $(L,m)\in\psublin/\poly$. Setting $t(n)=n^{O(1)}$ and $\ell(n)=n^{\varepsilon}$, Proposition \ref{nonuniform-ptimespace}(1) yields a nonuniform family $\{M_{n,l}\}_{n,l\in\nat}$ of $n^{O(1)}$-state $O(n^{\varepsilon})$-narrow 2afa's such that $M_{m(x),|x|}$ computes $L(x)$ in $|x|^{O(1)}$ time from each input $x$. Let $\ell(n)=c_2n\log^k{n}$. Since $|x|\leq \ell(m(x))$, we define a new 2afa $N_n$ as follows.
\begin{quote}
On input $x$, run $D_{|x|}$ on $x$ and obtain $1^{m(x)}$. If $m(x)=n$, then run $M_{n,|x|}$ on $x$; otherwise, reject $x$.
\end{quote}
This 2afa $N_{n}$ clearly solves the promise decision problem $(L_{n},\overline{L}_{n})$. This fact implies that $\LL\in 2\mathrm{A}_{narrow(n^{\varepsilon})}$, and therefore $\LL$ is in $2\tilde{\mathrm{A}}$.
\end{proof}

To the pair $(L,m)$ with $L=3\dstcon$ and $m=m_{ver}$, we can apply Proposition \ref{PsubLIN-to-2A}, because $3\mathcal{DSTCON}$ is induced from $(3\dstcon,m_{ver})$. As an immediate consequence of the proposition, we instantly obtain the following corollary.

\begin{corollary}\label{relation-3DSTCON}
$(3\dstcon,m_{ver})\in\psublin/\poly$ if and only if $3\mathcal{DSTCON}  \in  \bigcup_{\varepsilon\in[0,1)} 2\mathrm{A}_{narrow(n^{\varepsilon})}$.
\end{corollary}


Proposition \ref{construct-3DSTCON} yields a log-space computable function $g$ producing $\{M_n\}_{n\in\nat}$ for which each $M_n$ has $O(n\log{n})$ states, is both $3$-branching and simple, and solves $(3\dstcon_n,\overline{3\dstcon}_n)$ in $O(n|x|)$ time  for all ``valid'' inputs $x$ and rejects all ``invalid'' inputs in $O(n|x|)$ time. From this fact, we can conclude that $3\mathcal{DSTCON}$ belongs to $2\mathrm{qlinN}$.

\begin{lemma}\label{3DSTCON-by-2linN}
$3\mathcal{DSTCON}$ is in $2\mathrm{qlinN}$.
\end{lemma}

Combining Corollary \ref{relation-3DSTCON} and Lemma \ref{3DSTCON-by-2linN} with Theorem \ref{nonunif-general-3DSTCON}(2), we establish the following logical  equivalence.

\begin{proposition}\label{switch-2linN-3DSTCON}
$2\mathrm{qlinN} \subseteq \bigcup_{\varepsilon\in[0,1)} 2\mathrm{A}_{narrow(n^\varepsilon)}$ if and only if $3\mathcal{DSTCON} \in \bigcup_{\varepsilon\in[0,1)}2\mathrm{A}_{narrow(n^\varepsilon)}$.
\end{proposition}

\begin{proof}
As before, we abbreviate $\bigcup_{\varepsilon\in[0,1)} 2\mathrm{A}_{narrow(n^\varepsilon)}$ as $2\tilde{\mathrm{A}}$. Since  $3\mathcal{DSTCON}$ is induced from $(3\dstcon,m_{ver})$, we can identify $3\dstcon$ with $3\mathcal{DSTCON}$ in a natural way.

[Only If -- part] Assume that $2\mathrm{qlinN} \subseteq 2\tilde{\mathrm{A}}$. Since $3\mathcal{DSTCON} \in 2\mathrm{qlinN}$ by Lemma \ref{3DSTCON-by-2linN}, it immediately follows that   $3\mathcal{DSTCON} \in 2\tilde{\mathrm{A}}$.

[If -- part] Assume that $3\mathcal{DSTCON} \in 2\tilde{\mathrm{A}}$.
Corollary \ref{relation-3DSTCON} implies that $(3\dstcon,m_{ver})\in \psublin/\poly$; more specifically, $(3\dstcon,m_{ver})\in \ptimespace(m_{ver}(x)^{\varepsilon})/\poly$ for a certain constant $\varepsilon\in[0,1)$. By setting $\ell(n)=n^{\varepsilon}$, Theorem \ref{nonunif-general-3DSTCON}(2) provides a constant $s>0$ such that (*) for any $n$-state $c$-branching simple 2nfa, we can find an equivalent
$O(n^{\varepsilon})$-narrow 2afa of at most $O(n^{s})$ states.

Let $\LL=\{(L_n,\overline{L}_n)\}_{n\in\nat}$ be any family of promise decision problems in $2\mathrm{qlinN}$. Hereafter, we intend to prove that $\LL\in2\tilde{\mathrm{A}}$. Since $\LL\in2\mathrm{qlinN}$, we take two constants $c,k>0$ and a family $\{M_n\}_{n\in\nat}$ of $c$-branching simple 2nfa's $M_n$ with at most $cn\log^k{n}$ states solving $(L_n,\overline{L}_n)$ for all inputs.
By the above statement (*), for each $M_n$, there exists an equivalent $O((cn\log^k{n})^{\varepsilon})$-narrow 2afa $N_n$ of at most $(cn\log^k{n})^s+s$  states.  Note that there is a constant $\varepsilon'\in[0,1)$ satisfying both
$(cn\log^k{n})^s\leq n^{2s}$ and $(cn\log^k{n})^{\varepsilon}\leq n^{\varepsilon'}$ for all but finitely many numbers $n\in\nat$. This implies that $N_n$ is $O(n^{\varepsilon'})$-narrow and has $O(n^{2s})$ states.
Since $\{M_n\}_{n\in\nat}$ computes $\LL$, $\{N_{n}\}_{n\in\nat}$ computes $\LL$ as well. Therefore, we conclude that $\LL\in 2\tilde{\mathrm{A}}$, as requested.
\end{proof}

In the end, we will demonstrate that Statements (1) and (3) of Theorem \ref{3DSTCON-char-nonunif} are logically equivalent by combining Corollary \ref{relation-3DSTCON} and Proposition \ref{switch-2linN-3DSTCON}.

\begin{proofof}{Theorem \ref{3DSTCON-char-nonunif}(3)}
We write $2\tilde{\mathrm{A}}$ for $\bigcup_{\varepsilon\in[0,1)} 2\mathrm{A}_{narrow(n^\varepsilon)}$ as before. Corollary \ref{relation-3DSTCON} implies that nonuniform LSH fails iff $3\mathcal{DSTCON}\in 2\tilde{\mathrm{A}}$.
Proposition \ref{switch-2linN-3DSTCON} shows that
$2\mathrm{qlinN} \subseteq 2\tilde{\mathrm{A}}$ iff $3\mathcal{DSTCON} \in 2\tilde{\mathrm{A}}$. Therefore, the logical equivalence between Statements (1) and (3) of Theorem \ref{3DSTCON-char-nonunif} follows instantly.
\end{proofof}

\section{Case of Unary Finite Automata}\label{sec:unary-case}

We intend to shift our attention to families of unary promise decision problems and unary finite automata. Our goal is to prove Theorem \ref{uniform-unary-char} by exploring the expressive power of \emph{unarity} for certain promise decision problems.

Since the proof of Theorem \ref{uniform-unary-char} needs a unary version of $3\mathcal{DSTCON}$, we need to seek an appropriate \emph{unary encoding} of a degree-bounded subgraph of each complete directed graph $K_n$. For this purpose,  we use the following unary encoding scheme.
Fix $n\in\nat^{+}$. For any pair  $i,j\in[0,n-1]_{\integer}$ with $(i,j)\neq(0,0)$, let $p_{(i,j)}$ denote the $(i\cdot n+ j)$th prime number. Note that it is rather easy to decode $p_{(i,j)}$ to obtain a unique pair $(i,j)$. Given a degree-$3$ subgraph $G=(V,E)$ of $K_n$ with $V=\{0,1,2,\ldots,n-1\}$, the \emph{unary encoding} $\pair{G}_{unary}$ of $G$ is a string of the form $1^{e}$ with $e = \prod_{l=1}^{k}p_{(i_l,j_l)}$, where $E=\{(i_1,j_1),(i_2,j_2),\ldots,(i_k,j_k)\}\subseteq V^2$ with $k=|E|$.  Since $G$'s vertex has degree at most $3$, it follows that $k\leq 3n$.  It is also known that the $r$th prime number is at most $cr\log{r}$ for a certain constant $c>0$.  Since $i\cdot n+j\leq n^2$ for all pairs $i,j\in V$, it then follows that $p_{(i,j)}\leq cn^2\log{n^2} = 2cn^2\log{n}$. We thus conclude that $|\pair{G}_{unary}| = e \leq (c'n^2\log{n})^{3n}$, where $c'=2c$. Let $u3\mathcal{DSTCON} = \{(u3\dstcon_n,\overline{u3\dstcon}_n)\}_{n\in\nat}$ be defined as follows.

\s
Unary 3DSTCON of Size $n$ $(u3\dstcon_n,\overline{u3\dstcon}_n)$:
\renewcommand{\labelitemi}{$\circ$}
\begin{itemize}\vs{-2}
  \setlength{\topsep}{-2mm}%
  \setlength{\itemsep}{1mm}%
  \setlength{\parskip}{0cm}%

\item Instance: an encoding $\pair{G}_{unary}$ of a subgraph $G$ of $K_n$ with vertices of degree at most $3$.

\item Output: YES if there is a path from vertex $0$ to vertex $n-1$;  NO otherwise.
\end{itemize}

For the proof of Theorem \ref{uniform-unary-char}, we first need a unary version of Proposition \ref{construct-3DSTCON}. Recall that any output tape of finite automata is a semi-infinite write-only tape.

\begin{proposition}\label{unary-2nfa-DSTCON}
There exists a log-space computable function $g$ that, on each input $1^n$, produces an encoding $\pair{N_n}$ of $3$-branching simple unary 2nfa $N_n$ with $O(n^3\log{n})$ states that solves $(u3\dstcon_n,\overline{u3\dstcon}_n)$ in $O(n^2|x|)$ time and rejects all inputs outside of $\Sigma_n = u3\dstcon_n \cup \overline{u3\dstcon}_n$ in $O(n^2|x|)$ time, where $x$ indicates a symbolic input to $N_n$.
\end{proposition}

\begin{proof}
Given an index $n\in\nat^{+}$, we plan to construct an $O(n^3\log{n})$-state  3-branching simple unary 2nfa $M_n$ that solves $(u3\dstcon_n,\overline{u3\dstcon}_n)$ in $O(n^2|x|)$ time on all inputs $x$.
Let $x = \pair{G}_{unary}$ be any unary input given to $(u3\dstcon_n,\overline{u3\dstcon}_n)$, where $G=(V,E)$ with $V=[0,n-1]_{\integer}$ for $n=m_{ver}(x)$. On the input $x$, the desired $M_n$ works as follows.
\begin{quote}
Starting with the initial state $q_0$ scanning $\cent$, we always move a tape head rightward along a circular input tape. Initially, we choose the vertex $0$. Inductively, we follow an edge to one of the adjacent vertices as follows. Assume that we are currently at vertex $i$ at scanning $\dollar$. Choose a number $k\in\{1,2,3\}$ nondeterministically and step forward to $\cent$ since the tape is circular. By sweeping the tape from $\cent$ to $\dollar$, we search for the $k$th smallest index $j\in[0,n-1]_{\integer}$ for which $|x|$ is divisible by $p_{(i,j)}$, by incrementing $j$ from $0$ to $n-1$ and also checking whether or not $|x|$ is divisible by $p_{(i,j)}$. The divisibility can be checked by repeatedly sectioning $1^{|x|}$ into a block consisting of $p_{(i,j)}$ 1s. Continue searching more such $j$'s until we locate the $k$th such $j$. This process requires at most $n$ sweeping movements along the tape. If the vertex $n-1$ is reached, then we accept the input. Otherwise, when $n$ vertices are chosen without reaching the vertex $n-1$, we reject the input.
\end{quote}
For each fixed index $i$, we use elements in $\{i\}\cup \{(i,j,k,l)\mid 0\leq j< n,1\leq l\leq p_{(i,j)}, k\in[3]\}$ as inner states to implement the above procedure. Since $i$ ranges from $0$ to $n-1$, overall, $M_n$ uses at most $n\cdot O(n+ n^2\log{n}) = O(n^3\log{n})$ states. The running time of $M_n$ is $n\cdot O(n|x|) = O(n^2|x|)$.
\end{proof}

Since the length of a unary string $\pair{G_x}_{unary}$ in general is too large to handle within polynomially many steps in $m_{ver}(x)$, we need to consider a scaled-down version of $\pair{G_x}_{unary}$. Assuming that $G_x=(V,E)$ with $V=[0,m-1]_{\integer}$, and $E=\{(i_1,j_1),(i_2,j_2),\ldots,(i_k,j_k)\}$, we define a \emph{prime encoding} of $G_x$ as $bin_{s}(p_{(i_1,j_1)})\# bin_{s}(p_{(i_2,j_2)})\#\cdots \# bin_{s}(p_{(i_k,j_k)})$, where $s=\ceilings{\log{p_{(n-1,n-1)}}}$ (with $s\leq 4\log{n}$ for any sufficiently large $n$) and $k\leq 3n$.
We do not dare to fix the order of those prime numbers. Notationally, we write $\pair{G_x}_{prime}$ for such a prime encoding of $G_x$.
Note that $\sum_{t=1}^{k}p_{(i_t,j_t)}\leq (2cn^2\log{n})\cdot k \leq 6cn^3\log{n}$ for a certain constant $c>0$. It thus follows that $|\pair{G_x}_{prime}|\leq s\cdot k \leq 3n\cdot 4\log{n} = O(n\log{n})$.

Let $3\mathcal{DSTCON}_{prime}$ denote the promise problem obtained from $u3\mathcal{DSTCON}$ by replacing $\pair{G}_{unary}$ with $\pair{G}_{prime}$. The following lemma discusses how to transform $\pair{G}$ to $\pair{G}_{prime}$ for any degree-$3$ graph $G$.

\begin{lemma}\label{graph-binary-to-unary}
There exists a function $h$ that takes any input of the form $\pair{G}$ for a degree-$3$ subgraph $G$ of $K_n$ for a certain index $n\in\nat$ and outputs $\pair{G}_{prime}$. This function $h$ is computed by a certain $\dl$-uniform family of $n^{O(1)}$-state simple 2dfa's running in $n^{O(1)}\cdot O(|x|)$ time on all inputs $x$.
\end{lemma}

\begin{proof}
Fix $n\in\nat$ and let $G=(V,E)$ be any degree-$3$ subgraph of $K_n$ with $V=[0,n-1]_{\integer}$ and $E\subseteq V^2$. Let us recall from Section \ref{sec:automata-3DSTCON} that $\pair{G}$ is of the form $\pair{C_1\#^2 C_2\#^2 \cdots \#^2 C_n}$ with $C_i= binary(i)\# binary^*(j_{i,1}) \# binary^*(j_{i,2})\# binary^*(j_{i,3})$ for certain indices $j_{i,1},j_{i,2},j_{i,3} \in[0,n-1]_{\integer}\cup\{\bot\}$. The desired simple 2dfa that produces $\pair{G}_{prime}$ from $\pair{G}$ performs in the following way.
\begin{quote}
By sweeping an input tape repeatedly from the left to the right, we do the following. Assume that $E=\{(i_1,j_1),(i_2,j_2),\ldots,(i_k,j_k)\}$ with $k=|E|\leq 3n$. We then determine a set $S = \{p_{(i_1,j_1)},p_{(i_2,j_2)},\ldots,p_{(i_k,j_k)}\}$ of prime numbers. Let  $s=\ceilings{\log{p_{(n-1,n-1)}}}$.
Once we obtain the set $S$, we write $bin_{s}(p_{(i_1,j_1)})\# bin_{s}(p_{(i_2,j_2)})\#\cdots \# bin_{s}(p_{(i_k,j_k)})$
onto a write-only output tape to produce $\pair{G}_{prime}$.
\end{quote}
Since $|\pair{G}_{prime}|=O(n\log{n})$, we need $n^{O(1)}$ states to implement the above procedure. Moreover, the procedure clearly takes $n^{O(1)}\cdot O(|x|)$ steps.
\end{proof}

A key ingredient of the proof of Theorem \ref{uniform-unary-char} is the following lemma, which is inspired by the proof of \cite[Lemma 6]{KP15}. The lemma provides an effective method of constructing 2afa for $3\mathcal{DSTCON}_{prime}$ from simple 2afa's for $u3\mathcal{DSTCON}$.

\begin{lemma}\label{2afa-translate}
Let $s,t$ be polynomials and let $f:\nat\to\nat$ be any strictly increasing function. If $\{M_n\}_{n\in\nat}$ is an $\dl$-uniform family of $s(n)$-state  $f(n)$-narrow simple 2afa's solving $u3\mathcal{DSTCON}$ in $t(n)\cdot O(|x|)$ time, then there is another $\dl$-uniform family $\{P_n\}_{n\in\nat}$ of $O(s(n))$-state $O(f(n))$-narrow 2afa's that solves $3\mathcal{DSTCON}_{prime}$ in $O(t(n))\cdot O(|x|)$ time, where $x$ is a symbolic input.
\end{lemma}

\begin{proof}
An important observation is that, since $M_n$ is sweeping and end-branching, it behaves like a 1dfa while reading $\cent x$ for any input $x$. For ease of description, we partition a computation of $M_n$ into a series of \emph{sweeping rounds} so that, at each sweeping round, $M_n$ sweeps the input tape once from $\cent$ to $\dollar$.
Thus, we can take a number $d\in\nat^{+}$ and a series of $d$ sets $S_1,S_2,\ldots,S_d$ of inner states of $M_n$ such that, for each index $i\in[d]$, if  an input is sufficiently long and $M_n$ starts in a certain inner state, say, $q_1$ in $S_i$ at scanning $\cent$, then $M_n$ takes all inner states of $S_i$ sequentially (possibly with repetitions), and ends in a certain inner state, say, $q_{k_3}$ of $S_i$.
Since inputs are unary strings, $S_i$ is uniquely determined by $M_n$ and is of the form $\{q_1,q_2,\ldots,q_{k_2}\}$ for a certain index $k_2\geq1$.
We can list all inner states of $S_i$ used by $M_n$ as (*) a series  $(q_1,q_2,\ldots,q_{k_1-1}, (q_{k_1},q_{k_1+1},\cdots,q_{k_2})^r, q_{k_2+1},q_{k_2+2},\cdots,q_{k_3})$, where $k_1\geq1$, $k_2\geq k_1-1$, $k_3\geq k_2$, $(q_{k_1},\cdots q_{k_2})^r$ means the $r$ repetitions of the series $(q_{k_1},\cdots q_{k_2})$ and $(q_{k_2+1},\ldots,q_{k_3})$ is an initial segment\footnote{An initial segment of a series $(a_1,a_2,\ldots,a_n)$ is of the form $(a_1,a_2,\ldots,a_k)$ for a certain index $k\in[n]$.} of $(q_{k_1},\ldots,q_{k_2})$.
Note that all elements in $\{q_1,q_2, \cdots,q_{k_1}, \cdots,q_{k_2}\}$ must be distinct.

Next, let us define the desired 2afa $P_n$. Let $G$ be any subgraph of $K_n$ of degree at most $3$ and let $x$ denote $\pair{G}_{prime}$. Consider its associated unary encoding $\pair{G}_{unary}$ of the form $1^e$ for a certain number $e\in\nat^{+}$. By the definitions of $\pair{G}_{unary}$ and $\pair{G}_{prime}$, it follows that $e=\prod_{l=1}^{k}p_{(i_l,j_l)}$, provided that $\pair{G}_{prime} = bin_s(p_{(i_1,j_1)})\# bin_s(p_{(i_2,j_2)}) \# \cdots \# bin_s(p_{(i_k,j_k)})$ with $k\in\nat^{+}$ and $s=\ceilings{\log{p_{(n-1,n-1)}}}$. The 2afa $P_n$ works as follows.
\begin{quote}
We start with the same initial state of $M_n$ at scanning $\cent$. On the input $\pair{G}_{unary}$, we simulate all sweeping rounds of $M_n$ one by one in the following manner. At each sweeping round, let $q_1$ be an inner state of $M_n$ at scanning $\cent$. Assume that $M_n$ takes a series $(q_1,q_2,\ldots,q_{k_1-1}, (q_{k_1},q_{k_1+1},\cdots,q_{k_2})^r, q_{k_2+1},q_{k_2+2},\cdots,q_{k_3})$ of inner states, as explained as above.
Instead of moving the tape head along the input tape, we try to determine the value of $q_{k_3}$ as follows. Since $|\cent x\dollar| = (k_1-1) + r(k_2-k_1+1) + (k_3-k_2)$, we obtain $e+2=r(k_2-k_1+1)+(k_1-k_2-1)+k_3$; thus, $k_3= e+2\;(\mathrm{mod}\; k_2-k_1+1)$. Once $k_3$ is found, we can determine the inner state $q_{k_3}$ and then change the current inner state $q_1$ to $q_{k_3}$ by stepping to the right.
\end{quote}
Since $M_n$ is generated from $1^n$ in polynomial time using log space, we can determine $S_1,S_2,\ldots,S_d$ in log space by running $M_n$ on $1$s because we do not need to remember all inner states in each $S_i$. For each index $i\in[d]$, we can determine the aforementioned series (*) for $S_i$.
Thus, we can construct $P_n$ using log space in $n$. This concludes that $\{P_n\}_{n\in\nat}$ is $\dl$-uniform.

Overall, we need only $O(s(n))$ states to carry out the above procedure of $P_n$ in $O(t(n))\cdot O(|x|)$ time since $M_n$ runs in $t(n)\cdot O(|x|)$ time. Note that, when $M_n$ reaches $\dollar$, $P_n$ makes the same $\forall$- or $\exists$-move as $M_n$ does. Therefore, $P_n$ is of $O(f(n))$-narrow because so is $M_n$.
\end{proof}

Finally, we are ready to describe the proof of Theorem \ref{uniform-unary-char}.

\begin{proofof}{Theorem \ref{uniform-unary-char}}
Proposition \ref{2nfa-to-subgraphs} guarantees the existence of a function $g$ that changes $\pair{M}\# x$ for a  $c$-branching simple 2nfa $M$
into $\pair{G_x}$ for an appropriate degree-3 subgraph $G_x$ of $K_{d(n)}$ and an appropriate function $d(n)=O(n)$. Lemma \ref{graph-binary-to-unary} provides us with a function $h$ that transforms $\pair{G}$ to $\pair{G}_{prime}$ for any degree-3 subgraph $G$ of $K_n$.
It is important to note that $h$ and $g$ can be implemented by appropriate $\dl$-uniform families of $n^{O(1)}$-state simple 2dfa's running in $n^{O(1)}\cdot O(|x|)$ time.
By Proposition \ref{unary-2nfa-DSTCON}, we obtain a constant $e>0$ and an $\dl$-uniform family $\{D_n\}_{n\in\nat}$ of $3$-branching simple 2nfa's of at most $en^3\log{n}+e$ states, each $D_n$ of which solves $(u3\dstcon_n,\overline{u3\dstcon}_n)$ in $n^{O(1)}\cdot O(|x|)$ time and rejects  all inputs outside of $\Sigma_n = u3\dstcon_n \cup \overline{u3\dstcon}_n$, where $x$ is a symbolic input.

(1) Assume in particular that, for a certain fixed constant $\varepsilon\in[0,1)$, every $\dl$-uniform family of $3$-branching simple unary 2nfa's with at most $en^3\log{n}+e$ states can be converted into another $\dl$-uniform family of equivalent $n^{O(1)}$-state $O(n^{\varepsilon})$-narrow simple unary 2afa's.
By this assumption and Lemma \ref{2afa-translate}, from $\{D_n\}_{n\in\nat}$, we obtain an $\dl$-uniform family $\{P_n\}_{n\in\nat}$ of $O(n^{\varepsilon})$-narrow simple 2afa's with $n^{O(1)}$ states that solves $3\mathcal{DSTCON}_{prime}$  in $n^{O(1)}\cdot O(|x|)$ time for all inputs $x$.

To lead to the failure of LSH, it suffices to show that Statement (2) of Theorem \ref{3DSTCON-char-uniform} is satisfied for the case of $c=3$. Let $a>0$ be a constant and let $\{M_n\}_{n\in\nat}$ be an arbitrary $\dl$-uniform family of $3$-branching simple 2nfa's with at most $an\log{n}$ states.
Let us consider a 2afa $N_n$ that works as follows.
\begin{quote}
On input $x$, construct $\pair{M_n}\# x$, and apply $g$ to obtain $\pair{G_x}$. Note that $M_n$ accepts $x$ iff $\pair{G_x}\in3\dstcon_n$.
Next, apply $h$ to $\pair{G_x}$ and obtain $\pair{G_x}_{prime}$. Compute $d(n)$ ($=O(n)$). Run $P_{d(n)}$ on $\pair{G_x}_{prime}$.
\end{quote}
Since $P_{d(n)}$ is a 2afa, $N_n$ is also a 2afa.
This new 2afa $N_n$ has $n^{O(1)}$ states and it is also $O(n^{\varepsilon})$-narrow. Since $\{M_n\}_{n\in\nat}$ is $\dl$-uniform, by the definition, the family $\{N_n\}_{n\in\nat}$ is also $\dl$-uniform.
Theorem \ref{3DSTCON-char-uniform} then yields the desired consequence.

(2) Since Statement (1) implies the failure of LSH, it suffices to show that Statement (2) leads to Statement (1). For this purpose, we assume Statement (2) and follow an argument used in proving [3 $\Rightarrow$ 2] of the proof of Theorem \ref{uniform-3DSTCON-general}.

Let $c>0$ be any constant and take a constant $\varepsilon\in[0,1)$ ensured by Statement (2). Let $e>0$ be any constant and let $\{M_n\}_{n\in\nat}$ be any $\dl$-uniform family of $c$-branching simple 2nfa's with at most $en^3\log{n}+e$ states.
By the $\dl$-uniformity, we take a log-space DTM $D$ that produces $\pair{M_n}$ from $1^n$ for every index $n\in\nat$. By Statement (2), there is a log-space computable function $f$ for which, on each encoding of a $c$-branching simple unary 2nfa with at most $en^3\log{n}+e$ states, $f$ outputs an encoding of its equivalent $n^{O(1)}$-state $O(n^{\varepsilon})$-narrow simple unary 2afa. Consider the following procedure: we first run $D$ on input  $1^n$ to generate $\pair{M_n}$ and then apply $f$ to $\pair{M_n}$. We write $N_n$ for the resulted 2afa. It is not difficult to show that $\{N_n\}_{n\in\nat}$ is the desired family of 2afa's.
\end{proofof}

\section{Discussion and Open Problems}

The \emph{linear space hypothesis} (LSH) was initially proven to be a useful working hypothesis in the fields of $\nl$-search and $\nl$-optimization problems \cite{Yam17a,Yam17b}. A further study has been expected to seek more practical applications in other fields. In this work, we have looked for an exact characterization of LSH in automata theory and this result has shed clear light on the essential meaning of the hypothesis from an automata-theoretic  viewpoint. A key to our work is,  as noted in Section \ref{sec:families-languages},  the discovery of a close connection between a parameterized decision problem and a family of promise decision problems. This discovery leads us to Theorems \ref{3DSTCON-char-uniform}--\ref{3DSTCON-char-nonunif}, which has established a close connection between LSH and state complexity of transforming restricted 2nfa's into restricted
2afa's. In the past literature, the state complexity was shown to be useful to characterize a few complexity-theoretical issues; for example, the $\dl=\nl$ problem \cite{BL77,SS78} and the $\nl\subseteq\dl/\poly$ problem  \cite{Kap14,KP15}. Our result has given an additional evidence to support the usefulness of the state complexity of automata. Another important contribution of this work is to have introduced a nonuniform variant of LSH and have demonstrated a nonuniform variant of the aforementioned characterization of LSH in terms of nonuniform state complexity.

There are a number of interesting questions left unsolved in this work. We wish  to list some of these unsolved questions for a future study along the line of LSH and state complexity of finite automata.

\renewcommand{\labelitemi}{$\circ$}
\begin{enumerate}\vs{-2}
  \setlength{\topsep}{-2mm}%
  \setlength{\itemsep}{1mm}%
  \setlength{\parskip}{0cm}%

\item Our ultimate goal is to prove or disprove LSH and its nonuniform variant. It is not immediately clear, nonetheless, that this goal is easier or more difficult to achieve than solving the $\dl=\nl$ problem.

\item Improve Proposition \ref{Barnes-translate} by determining the exact state complexity of transforming an $n$-state simple 2nfa to an equivalent narrow 2afa with help of neither Theorem \ref{uniform-3DSTCON-general} nor the result of Barnes et al.'s \cite{BBRS98}.

\item The statements of Theorems \ref{3DSTCON-char-uniform}--\ref{3DSTCON-char-nonunif} associated with the conversions of two types of finite automata are quite complicated. Provide much simpler characterizations.

\item It is still open whether $2\mathrm{qlinN}$ in Theorem \ref{3DSTCON-char-nonunif}(3)  can be replaced by $2\mathrm{N}$ or even $2\mathrm{N}/\poly$ (see \cite{Kap12} for their definitions). This is somewhat related to the question of whether we can replace $2\mathrm{SAT}_3$ in the definition of LSH by $2\mathrm{SAT}$ \cite{Yam17a}; if we can answer this question positively, then LSH is simply rephrased as $\nl\nsubseteq \psublin$. Determine whether or not the replacement of $\mathrm{2qlinN}$ by $\mathrm{2N}$ or even $\mathrm{2N}/\poly$ is possible.

\item At this moment, we cannot assert that the failure of LSH derives Statements (1)--(2) of Theorem \ref{uniform-unary-char}. We also do not know whether  the simplicity requirement of ``simple unary 2afa'' in the theorem can be replaced by ``unary 2afa''. Settle down these points and establish an exact characterization for unary automata.

\item It is known in \cite{Yam17a,Yam17b} that $\psublin$ is closed under  \emph{sub-linear-space reduction family Turing reductions} (or SLRF-T-reductions). It is rather easy to define a nonuniform version of SLRF-T-reductions. Find a natural nonuniform state complexity class that is closed under SLRF-reductions. For example, is $\mathrm{2N}$ or $\mathrm{2N}/\poly$ closed under such reductions?
\end{enumerate}


\let\oldbibliography\thebibliography
\renewcommand{\thebibliography}[1]{%
  \oldbibliography{#1}%
  \setlength{\itemsep}{0pt}%
}
\bibliographystyle{plain}

\begin{thebibliography}{Gur91}


\bibitem{ACL+14}
E. Allender, S. Chen, T. Lou, P. A. Papakonstantinou, and B. Tang. Width-parameterized SAT: time-space tradeoffs. Thoery of Computing 10 (2014) 297--339.

\bibitem{BBRS98}
G. Barnes, J. F. Buss, W. L. Ruzzo, and B. Schieber. A sublinear space, polynomial time algorithm for directed s-t connectivity. {SIAM J. Comput.} 27 (1998) 1273--1282.

\bibitem{BL77}
P. Berman and A. Lingas. On complexity of regular languages in terms of finite automata. Report 304, Institute of Computer Science, Polish Academy of Science, Warsaw, 1977.

\bibitem{CKS81}
A. Chandra, D. Kozen, and L. Stockmeyer. Alternation. J. of ACM 28 (1981) 114--133.

\bibitem{CT15}
D. Chakraborty and R. Tewari. Simultaneous time-space upper bounds for red-blue path problem in planar DAGs. In Proc. of the 9th International Workshop on Algorithms and Computation (WALCOM 2015), Lecture Notes in Computer Science, Springer, vol. 8973, pp. 258--269, 2015.

\bibitem{GGP14}
V. Geffert, B. Guillon, and G. Pighizzini. Two-way automata making choices only at the endmarkers. Inf. Comput. 239 (2014) 71--86.

\bibitem{GMP03}
V. Geffert, C. Mereghetti, and G. Pighizzini. Converting two-way nondeterministic automata into simpler automata. Theor. Comput. Sci. 295 (2003) 189--203.

\bibitem{GO14}
V. Geffert and A. Okhotin. Transforming two-way alternating finite automata to one-way nondeterminsitic automata. In the Proc. of the 39th Mathematical Foundations of Computer Science (MFCS 2014), Lecture Notes in Computer Science, Springer, vol. 8634 (part I), pp. 291--302, 2014.

\bibitem{GP11}
V. Geffert and G. Pighizzini. Two-way unary automata versus logarithmic space. Inform. Comput. 209 (2011) 1016--1025.

\bibitem{Kap09}
C. A. Kapoutsis. Size complexity of two-way finite automata. In the Proc. of the 13th International Conference on Developments in Language Theory (DLT 2009), Lecture Notes in Computer Science, Springer, vol. 5583, pp. 47--66, 2009.

\bibitem{Kap12}
C. A. Kapoutsis. Minicomplexity. J. Automat. Lang. Combin. 17 (2012) 205--224.

\bibitem{Kap14}
C. A. Kapoutsis. Two-way automata versus logarithmic space. Theory Comput. Syst. 55 (2014) 421--447.

\bibitem{KP15}
C. A. Kapoutsis and G. Pighizzini. Two-way automata characterizations of L/poly versus NL. Theory Comput. Syst. 56 (2015) 662--685.

\bibitem{KL81}
R. M. Karp and R. Lipton, Turing machines that take advice.
Enseig. Math. 28 (1982) 191--209.

\bibitem{SS78}
W. J. Sakoda and M. Sipser. Nondeterminism and the size of two-way finite automata. In the Proc. of the 10th ACM Symposium on Theory of Computing (STOC'78), pp. 275--286, 1978.

\bibitem{Yam17a}
T. Yamakami. The 2CNF Boolean formula satsifiability problem and the linear space hypothesis. In the Proc. of the 42nd International Symposium on Mathematical Foundations of Computer Science (MFCS 2017),
Schloss Dagstuhl,  Leibniz International Proceedings in Informatics (LIPIcs), vol. 83, pp. 62:1--62:14, 2017. A complete and correct version is available  at arXiv:1709.10453.

\bibitem{Yam17b}
T. Yamakami. Parameterized graph connectivity and polynomial-time sub-linear-space short reductions (preliminary report). In the Proc. of the 11th International Workshop on Reachability Problems (RP 2017), Lecture Notes in Computer Science, Springer, vol. 10506, pp. 176--191, 2017.

\end{thebibliography}

\end{document}